\documentclass[aps,prl,superscriptaddress,citeautoscript,twocolumn,reprint,longbibliography,floatfix]{revtex4-2}
\usepackage{graphicx,amssymb,amsmath,epsf,bm,cprotect,comment,physics}
\epsfclipon

\newcommand{\expct}[1]{\langle{#1}\rangle}

\newcommand{\prt}[2]{\frac{\partial{#1}}{\partial{#2}}}

\renewcommand{\eqref}[1]{Eq.\,(\ref{#1})}
\newcommand{\eqsref}[1]{Eqs.\,(\ref{#1})}
\newcommand{\pref}[1]{(\ref{#1})}
\newcommand{\figref}[1]{Fig.\,\ref{#1}}
\newcommand{\figsref}[1]{Figs.\,\ref{#1}}
\newcommand{\tblref}[1]{Table\,\ref{#1}}

\newcommand{\supfigref}[1]{Fig.\,\ref{#1}}
\newcommand{\supfigsref}[1]{Figs.\,\ref{#1}}

\newcommand{\suptxtref}[1]{Supplemental Text~#1}

\newcommand{\etal}{\textit{et al}.\ }
\newcommand{\fKPZ}{f_\mathrm{KPZ}}

\usepackage{xr} 
\makeatletter
\newcommand*{\addFileDependency}[1]{
  \typeout{(#1)}
  \@addtofilelist{#1}
  \IfFileExists{#1}{}{\typeout{No file #1.}}
}
\makeatother

\newcommand*{\myexternaldocument}[1]{
    \externaldocument[S-]{build/#1}  
    \addFileDependency{#1.tex}
    \addFileDependency{build/#1.aux}  
}
\myexternaldocument{suppl}

\begin{document}

\title{Partial yet definite emergence of the Kardar-Parisi-Zhang class in isotropic spin chains}

\author{Kazumasa A. Takeuchi}
\email{kat@kaztake.org}
\affiliation{Department of Physics,\! The University of Tokyo,\! 7-3-1 Hongo,\! Bunkyo-ku,\! Tokyo 113-0033,\! Japan}%
\affiliation{Institute for Physics of Intelligence,\! The University of Tokyo,\! 7-3-1 Hongo,\! Bunkyo-ku,\! Tokyo 113-0033,\! Japan}%

\author{Kazuaki Takasan}
\affiliation{Department of Physics,\! The University of Tokyo,\! 7-3-1 Hongo,\! Bunkyo-ku,\! Tokyo 113-0033,\! Japan}%

\author{Ofer Busani}
\affiliation{School of Mathematics, University of Edinburgh, James Clerk Maxwell Building, Peter Guthrie Tait Road, Edinburgh, EH9 3FD, UK}%

\author{Patrik L. Ferrari}
\affiliation{Institute for Applied Mathematics, Bonn University, Endenicher Allee 60, 53115 Bonn, Germany}%


\author{Romain Vasseur}
\affiliation{Department of Theoretical Physics, University of Geneva, 24 quai Ernest-Ansermet, 1211 Gen\`eve, Switzerland}

\affiliation{Department of Physics, University of Massachusetts, Amherst, MA 01003, USA}%

\author{Jacopo De Nardis}
\affiliation{Laboratoire de Physique Th\'eorique et Mod\'elisation, CNRS UMR 8089,
CY Cergy Paris Universit\'e, 95302 Cergy-Pontoise Cedex, France}%

\date{\today}

\begin{abstract}
Integrable spin chains with a continuous non-Abelian symmetry, such as the one-dimensional isotropic Heisenberg model, show superdiffusive transport with little theoretical understanding. Although recent studies reported a surprising connection to the Kardar-Parisi-Zhang (KPZ) universality class in that case, this view was most recently questioned by discrepancies in full counting statistics. Here, by combining extensive numerical simulations of classical and quantum integrable isotropic spin chains with a framework developed by exact studies of the KPZ class, we characterize various two-point quantities that remain hitherto unexplored in spin chains, and find full agreement with KPZ scaling laws without adjustable parameters. This establishes the partial emergence of the KPZ class in integrable isotropic spin chains. Moreover, we reveal that the KPZ scaling laws are intact in the presence of an energy current, under the appropriate Galilean boost required by the propagation of spacetime correlation.
\end{abstract}

\maketitle

Characterizing transport properties of quantum many-body systems, in particular those of integrable systems with non-diffusive transport, is a longstanding objective of condensed matter physics.
Integrability typically results in ballistic transport, as successfully described by the framework of the generalized hydrodynamics \cite{CastroAlvaredo.etal-PRX2016,Bertini.etal-PRL2016,Doyon-SPPLN2020}, but it is faced with challenges when ballistic contributions are canceled by symmetry or other mechanisms \cite{Bulchandani.etal-JSM2021,Gopalakrishnan.Vasseur-ARCMP2024}.
Paradigmatic is the situation with a continuous non-Abelian symmetry, in particular the isotropic Heisenberg spin chain, which was reported to show superdiffusive transport with characteristic length $\xi(t) \sim t^{2/3}$ \cite{Fabricius.McCoy-PRB1998,Ljubotina.etal-NC2017,GV2019,PhysRevLett.123.186601}.
Surprisingly, this superdiffusive exponent was associated with an apparently unrelated universality class established mainly for classical non-equilibrium systems, namely the Kardar-Parisi-Zhang (KPZ) universality class for fluctuations of growing interfaces and related phenomena \cite{[{For reviews, see, e.g., }]Takeuchi-PA2018,*Corwin-RMTA2012,*Prolhac-a2024}.
Key evidence \cite{Ljubotina.etal-PRL2019,Weiner.etal-PRB2020} was the precise agreement of the equilibrium two-point spin correlation function with Pr\"ahofer and Spohn's exact solution for the KPZ class \cite{Prahofer.Spohn-JSP2004}, often denoted by $\fKPZ(\cdot)$.
On the one hand, this alleged manifestation of the KPZ class is deemed universal \cite{Ilievski.etal-PRX2021,Bulchandani.etal-JSM2021,Gopalakrishnan.Vasseur-ARCMP2024}, as confirmed in various isotropic integrable spin chains, whether quantum \cite{Ye.etal-PRL2022} or classical \cite{Das.etal-PRE2019}, and also supported by a few experimental investigations \cite{Scheie.etal-NP2021,Wei.etal-S2022}.
On the other hand, it is clear from the symmetry of spins that the magnetization transfer (integrated spin current) must show a symmetric distribution, unless the symmetry is explicitly broken by the initial condition or an external field \cite{Krajnik.etal-PRL2022}, while for KPZ the corresponding quantity, namely the interface height increment, is intrinsically asymmetric \cite{Takeuchi-PA2018}.
Furthermore, recent computational studies \cite{Krajnik.etal-PRL2024,Rosenberg.etal-S2024} revealed 
discrepancies beyond the symmetry difference, notably in the kurtosis.
These findings led the authors to set aside the possibility that spin transport in such systems may be described by the KPZ class, instead suggesting the need for a new universality class
\cite{Krajnik.etal-PRL2024,Rosenberg.etal-S2024}.
After all, all pieces of evidence for KPZ reported so far have been rather weak, being the scaling exponents, which are simple rational numbers such as $2/3$, and the agreement with Pr\"ahofer and Spohn's solution $\fKPZ(\cdot)$, which has been compared with arbitrarily fitted scaling coefficients.
Serious doubt is cast on the relevance of the KPZ class in this context, or more broadly in characterizing transport properties of a class of quantum many-body systems.

Here we clarify the fate of the KPZ universality in isotropic integrable spin chains, both classical and quantum.
First we remark that the deep body of knowledge gained by mathematical studies on the 1D KPZ class \cite{Takeuchi-PA2018} has not been fully utilized.
It dictates, for example, the mutual relation between scaling coefficients.
They contain universal quantifiers, which are lost if treated as free fitting parameters.
Moreover, the Pr\"ahofer-Spohn function is not the only two-point correlator with an exact solution \cite{Takeuchi-PA2018}; other two-point functions, such as the equal-time spatial correlator \cite{Quastel.Remenik-Inbook2014} and equal-position two-time correlator \cite{Ferrari.Spohn-SIG2016,Ferrari.Occelli-MPAG2019} have also been dealt with.
The purpose of the present Letter is to make full use of these results to carry out a comprehensive test of the KPZ universality in isotropic integrable spin chains.
Our results on various two-point functions hitherto unexplored in spin chains reveal that the KPZ class is indeed relevant in spin chains, yet it describes two-point functions only, hence the partial emergence of the KPZ class.

We study both classical and quantum integrable isotropic spin chains.
For the classical case, we use Krajnik and Prosen's model \cite{Krajnik.Prosen-JSP2020,Krajnik.etal-SPP2021,Krajnik.etal-PRL2022} based on the lattice Landau-Lifshitz magnet \cite{Lakshmanan-PTRSA2011}, which we shall call the KPLL model.
It is an integrable variant of the lattice Landau-Lifshitz model defined on a brick-layer space-time lattice (see \suptxtref{1} and \supfigref{S-figS:LL} \cite{suppl} for the complete definition), which converges to the Ishimori chain \cite{Ishimori-JPSJ1982}
\begin{equation}
\prt{\bm{S}_j}{t} = \frac{\bm{S}_j \times \bm{S}_{j-1}}{1+\bm{S}_j \cdot \bm{S}_{j-1}} + \frac{\bm{S}_j \times \bm{S}_{j+1}}{1+\bm{S}_j \cdot \bm{S}_{j+1}},
\label{eq:LL}
\end{equation}
in the continuous time limit \cite{Krajnik.Prosen-JSP2020,Krajnik.etal-SPP2021}.
For the quantum case, we study the isotropic Heisenberg chain, which is a representative integrable model \cite{Takahashi_book} defined by
\begin{equation}
\hat{H} = \sum_{j} \bm{\hat{S}}_j \cdot \bm{\hat{S}}_{j+1}, \label{eq:Heisenberg}
\end{equation}
with spin-1/2 operator $\bm{\hat{S}}_j=(\hat{S}_j^x, \hat{S}_j^y, \hat{S}_j^z)$.
In the following, we use the classical KPLL model to realize large-scale simulations for inspecting supposedly universal statistical properties of isotropic integrable spin chains, which are then confirmed by the quantum Heisenberg simulations.
For the KPLL model, unless otherwise stated, we started from infinite-temperature equilibrium states and obtained $N=10^4$ independent realizations with system size $L=40,000$ and the periodic boundary condition, with time step $0.1$.
The $z$-component of the spins, $S_j^z(t)$, is our magnetization field, denoted by $m(x,t)$ with $x=j$ hereafter whenever appropriate.
Another quantity of interest is the integrated spin current, or the magnetization transfer, $h(x,t) \equiv \int_0^t J(x,t')dt'$, with spin current $J(x,t)$. 
The magnetization transfer $h(x,t)$ corresponds to the height increment of the growing interfaces, which is central in the studies of the KPZ class.

\begin{figure}[tb]
\centering
\includegraphics[width=\hsize]{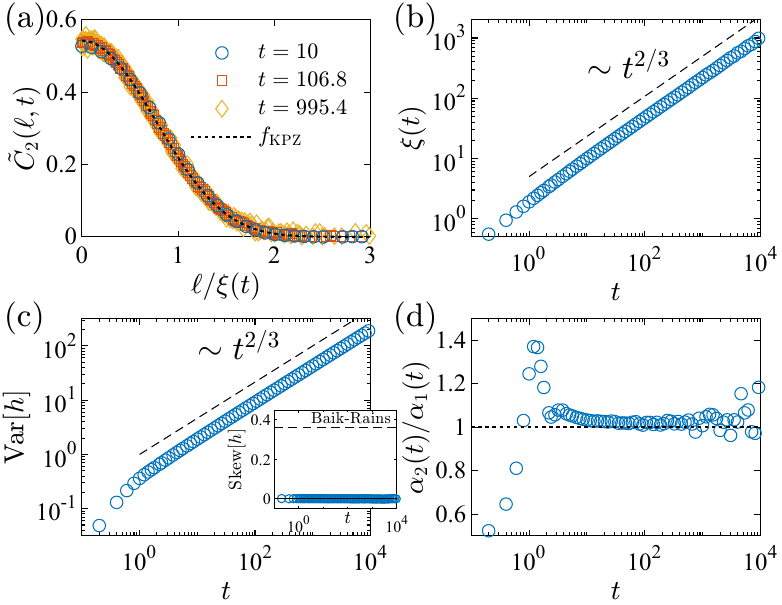}
\caption{
The two-point function and the magnetization transfer cumulants for the KPLL model.
(a) Rescaled two-point function $\tilde{C}_2(\ell,t) = \frac{\xi(t)}{\Omega}C_2(\ell,t)$ against $\ell/\xi(t)$, compared with the Pr\"{a}hofer-Spohn solution $\fKPZ(\cdot)$.
(b) Correlation length $\xi(t)$.
(c) Variance (main panel) and skewness (inset) of the magnetization transfer $h$. The dashed line in the inset indicates the skewness of the Baik-Rains distribution.
(d) Ratio of $\alpha_1(t)$ [from \eqref{eq:C2}] and $\alpha_2(t)$ [from \eqref{eq:var}].
}
\label{fig1}
\end{figure}

First we verify KPZ behavior of the KPLL model through the standard quantities.
Figure\,\ref{fig1}(a) displays the two-point function $C_2(\ell,t) \equiv \expct{m(x+\ell,t)m(x,0)}$, showing agreement with the Pr\"{a}hofer-Spohn exact solution $\fKPZ(\cdot)$.
Here the normalized function $\tilde{C}_2(\ell,t) \equiv \frac{\xi(t)}{\Omega}C_2(\ell,t)$ is shown, where $\Omega \equiv \int C_2(\ell,t)d\ell$ is conserved as a result of the conservation of the total magnetization $\int m(x,t)dx$, and $\xi(t)$ is the correlation length determined by $\frac{1}{\Omega}\int \ell^2 C_2(\ell,t) d\ell = \sigma^2 \xi(t)^2$ with $\sigma^2 \equiv \int u^2 \fKPZ(u) du \approx 0.51$.
This correlation length is confirmed to show the characteristic power law $\xi(t) \sim t^{2/3}$ of the KPZ class [\figref{fig1}(b)].
We also measure the variance of the magnetization transfer and find the characteristic KPZ growth, $\mathrm{Var}[h(x,t)] \sim t^{2/3}$ [\figref{fig1}(c)].
On the other hand, 
the skewness is zero and far from the value for the Baik-Rains distribution \cite{Baik.Rains-JSP2000} expected for the KPZ stationary state (inset).
Although the data shown so far are reproduction of known results \cite{Krajnik.etal-PRL2022,Krajnik.etal-PRL2024}, we can scrutinize nontrivial relationship underlying these quantities.
According to KPZ scaling laws \cite{Takeuchi-PA2018,Prahofer.Spohn-JSP2004,Iwatsuka.etal-PRL2020}, we have
\begin{gather}
C_2(\ell,t) \simeq \frac{2\alpha t^{2/3}}{\xi(t)^2} \fKPZ\left(\frac{\ell}{\xi(t)}\right),  \label{eq:C2} \\
\mathrm{Var}[h(x,t)] \simeq \alpha t^{2/3}\mathrm{Var[BR]},  \label{eq:var}
\end{gather}
where $\mathrm{Var[BR]} \approx 1.15$ is the variance of the Baik-Rains distribution and $\alpha$ is a coefficient.
Since these equations are not guaranteed to describe spin chains, here we evaluate $\alpha$ from data of $C_2(x,t)$ and $\mathrm{Var}[h(x,t)]$ independently and denote them by $\alpha_1(t)$ and $\alpha_2(t)$, respectively.
Then, remarkably, we find $\alpha_1(t) = \alpha_2(t)$ [\figref{fig1}(d)], substantiating the validity of the KPZ scaling laws \pref{eq:C2} and \pref{eq:var} in spin chains.
Note that \eqref{eq:var} includes the Baik-Rains variance,
even though the Baik-Rains distribution does not appear in spin chains.

\begin{figure}[tb]
\centering
\includegraphics[width=\hsize]{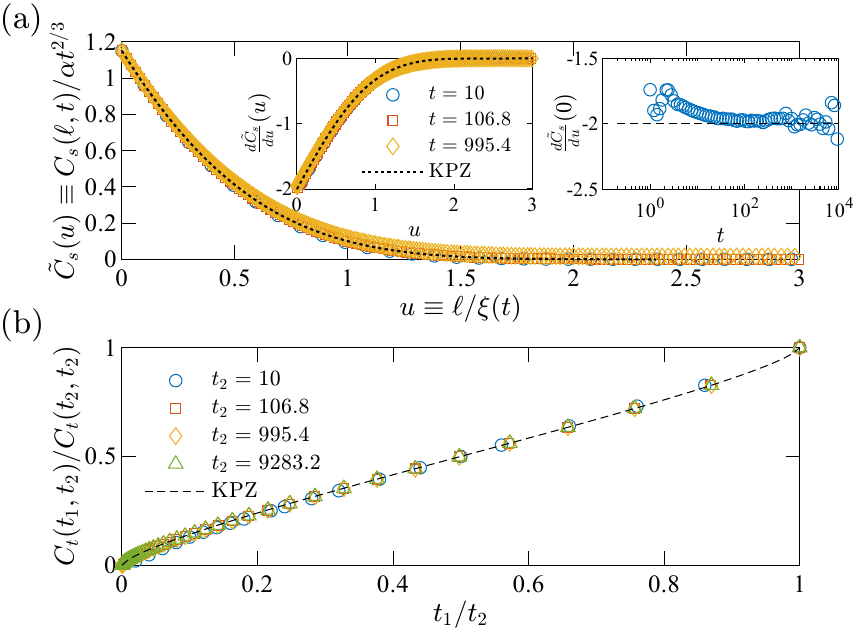}
\caption{
Spatial (a) and temporal (b) correlation functions of the magnetization transfer for the KPLL model.
(a) The rescaled spatial correlator $\tilde{C}_s(u) = C_s(\ell,t)/\alpha t^{2/3}$ (main panel) and its slope $\frac{d\tilde{C}_s}{du}(u)$ (left inset) against $u =\ell/\xi(t)$.
The dashed lines show the curves for the KPZ class, obtained by TASEP simulations.
The right inset compares the slope $\frac{d\tilde{C}_s}{du}(0)$ at $u=0$ with our exact result for the KPZ class, $\frac{d\tilde{C}_s}{du}(0)=-2$.
(b) The rescaled temporal correlator $C_t(t_1,t_2)/C_t(t_2,t_2)$ against $t_1/t_2$, compared with the Ferrari-Spohn solution [\eqref{eq:FS}] for the KPZ class.
}
\label{fig2}
\end{figure}

We further test the validity of KPZ scaling laws through other two-point quantities.
First we study the spatial correlation of the magnetization transfer:
\begin{equation}
C_s(\ell,t) \equiv \expct{h(x,t)h(x+\ell,t)} - \expct{h(x,t)}^2.
\end{equation}
Figure\,\ref{fig2}(a) shows it in the rescaled units, $\tilde{C}_s(u) \equiv C_s(\ell,t)/\alpha t^{2/3}$ against $u \equiv \ell/\xi(t)$.
For the KPZ class, the multi-point equal-time height correlation has been characterized intensively and described in terms of a family of stochastic processes called the Airy processes \cite{Prahofer.Spohn-JSP2002,Sasamoto-JPA2005,Borodin.etal-JSP2007,Widom-JSP2004,Basu.etal-CMP2023,Quastel.Remenik-Inbook2014}.
For the stationary state, a process called the Airy$_\mathrm{stat}$ process has been considered \cite{Baik.etal-CPAM2010} (see also a review \cite{Quastel.Remenik-Inbook2014}), but it describes the height measured in the absolute frame (say, $h_0(x,t)$) instead of the height increment $h(x,t) = h_0(x,t)-h_0(x,0)$ considered here.
We therefore introduce here the limiting process $\mathcal{A}_0(u)$ for the height increment $h(x,t)$, in other words the stationary version of the Airy$_\mathrm{stat}$ process, and call it the Airy$_0$ process.
We evaluate the covariance of the Airy$_0$ process $\mathcal{C}_0(u) \equiv \expct{\mathcal{A}_0(u)\mathcal{A}_0(0)}$ by numerical simulations of the totally asymmetric simple exclusion process (TASEP), a representative model in the KPZ class, and find it in excellent agreement with the data for the KPLL model [\figref{fig2}(a)].
Furthermore, we consider $\mathcal{C}_0(u)$ for small $u$ analytically and prove $\frac{d\mathcal{C}_0}{du}(0)=-2$ (see \suptxtref{2} \cite{suppl}).
This is confirmed by our data for both the KPLL model and the TASEP [insets of \figref{fig2}(a)].
Finally, we also investigate the temporal correlation of the magnetization transfer
\begin{equation}
C_t(t_1, t_2) \equiv \expct{h(x,t_1)h(x,t_2)} - \expct{h(x,t_1)}\expct{h(x,t_2)}.
\end{equation}
The results in \figref{fig2}(b) show excellent agreement with the exact solution for the KPZ class obtained by Ferrari and Spohn \cite{Ferrari.Spohn-SIG2016,Ferrari.Occelli-MPAG2019}:
\begin{equation}
 \frac{C_t(t_1, t_2)}{C_t(t_2, t_2)} \simeq \frac{1}{2}\left[1 + \tau^{2/3} - (1-\tau)^{2/3}\right],  \label{eq:FS}
\end{equation}
 with $\tau \equiv t_1/t_2$.
Theoretically, \eqref{eq:FS} is for $t_1,t_2 \to \infty$ with a fixed $\tau$, but the data converge remarkably fast.

\begin{figure}[tb]
\centering
\includegraphics[width=\hsize]{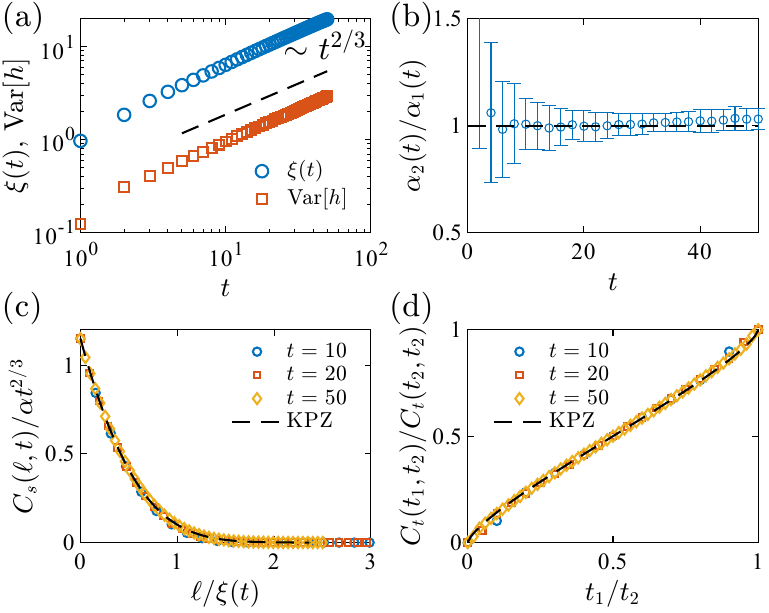}
\caption{
Results for the quantum Heisenberg model, with system size $L=100$ and maximum bond dimension $\chi=1600$ (see End Matter for details).
(a) Correlation length $\xi(t)$ and variance of the magnetization transfer, Var[$h$].
(b) Ratio $\alpha_2(t)/\alpha_1(t)$. See End Matter for the evaluation of the error bars.
(c,d) Spatial (c) and temporal (d) correlation functions of the magnetization transfer. 
}
\label{figQ}
\end{figure}

The results obtained so far for the classical KPLL model verified the validity of the KPZ scaling laws in various two-point quantities.
Remarkably, we confirm all these results in simulations of the quantum Heisenberg model too (\figref{figQ}) without any adjustable parameter. This establishes the universality of the results, encompassing both the quantum and classical worlds. The details of the simulation and the results of the quantum model are presented in End Matter.

Now we test the robustness of our findings under different situations, again using the classical KPLL model.
First, we consider the case with a non-vanishing energy current.
This is particularly tempting in view of the hydrodynamic description proposed by De Nardis \etal \cite{DeNardis.etal-PRL2023}, which predicts that left-moving and right-moving giant quasiparticles contribute equally to the magnetization, and this is why the distribution of $h$ becomes symmetric.
Therefore, it is important to clarify what happens if the left-right symmetry is broken, e.g., by the presence of a finite energy current.
We prepared such an initial condition by Monte Carlo sampling, using the statistical weight $\propto e^{-\lambda J_E}$ with total energy current $J_E \equiv -\sum_j \bm{S}_j \cdot (\bm{S}_{j+1} \times \bm{S}_{j+2})$ 
\footnote{
Since the exact expression of the energy current for the KPLL model is unknown, here we borrow the expression for the Heisenberg model
\cite{Zotos.etal-PRB1997}, which is expected to have a significant overlap with the true energy current of the KPLL model.
} and $\lambda=-1$.
Thereby, we indeed realize a situation where the energy current reaches a constant finite value after a short transient [\figref{fig3}(a) inset].

\begin{figure}[tb]
\centering
\includegraphics[width=\hsize]{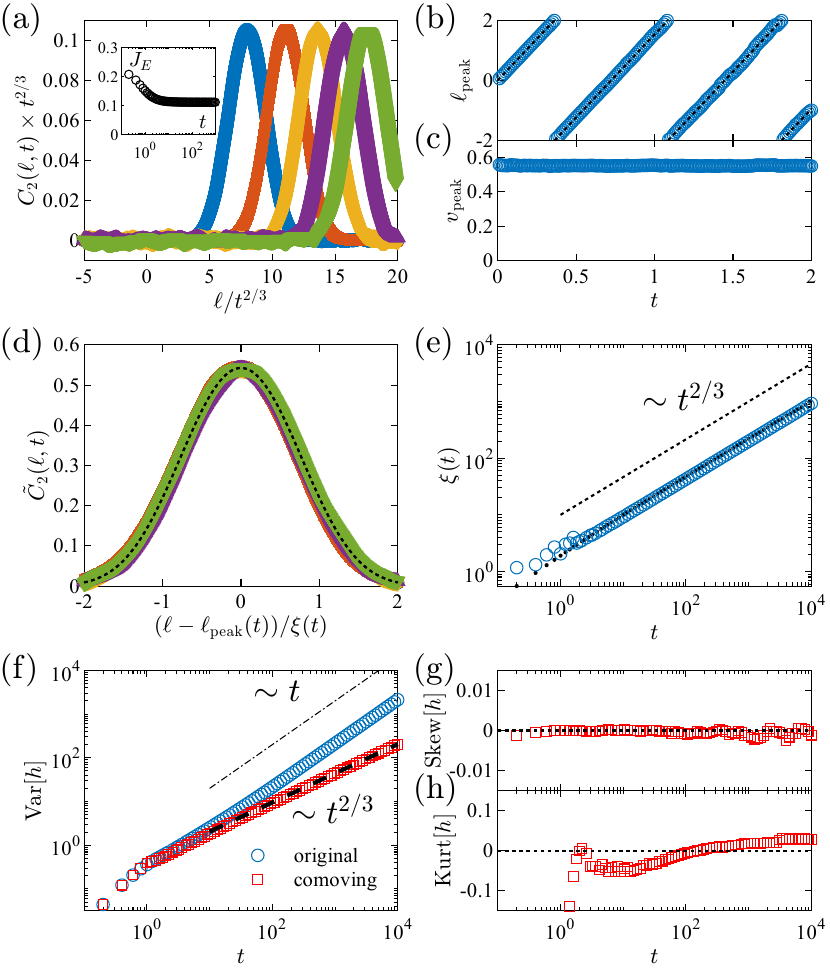}
\caption{
Results for KPLL with a finite energy current.
(a) Two-point function $C_2(\ell,t)/t^{2/3}$ against $\ell/\xi(t)$ for different times, $t = 3000, 8000, 15000, 23000, 32000$ from left to right. Data smoothed by the locally weighted scatterplot smoothing method are displayed.
Inset: total energy current $J_E(t)$.
(b)(c) The location and the velocity of the peak of $C_2(\ell,t)$, $\ell_\mathrm{peak}(t)$ and $v_\mathrm{peak}$, respectively. The dashed line in (b) shows $\ell_\mathrm{peak}(t) = vt$ with $v = 0.5523$, wrapped by the periodic boundary. 
(d) Rescaled two-point function $\tilde{C}_2(\ell,t)$ centered at $\ell = \ell_\mathrm{peak}(t)$ (symbols, same colors as (a)), compared with the Pr\"{a}hofer-Spohn solution $\fKPZ(\cdot)$ (dashed line).
(e) Correlation length $\xi(t)$. The black dots are the data for the case without energy current, shown in \figref{fig1}(b).
(f) Variance of the magnetization transfer $h(x,t)$, measured in the original and comoving frames (blue circles and red squares, respectively).
(g)(h) Skewness and kurtosis of the magnetization transfer $h(x,t)$ in the comoving frame.
The values for the Baik-Rains distribution are $0.359$ and $0.289$, respectively \cite{Prahofer.Spohn-PRL2000}, which are far from the data.
}
\label{fig3}
\end{figure}

Figure\,\ref{fig3}(a) shows the two-point function $\tilde{C}_2(\ell,t) = \frac{\xi(t)}{\Omega}C_2(\ell,t)$ in this case.
Interestingly, now we find the peak position of the correlation function moving at a constant velocity, $\ell_\mathrm{peak} = v_\mathrm{peak}t$ with $v_\mathrm{peak} = 0.5523$ [\figref{fig3}(b)(c)].
Apart from this, the form of the two-point function turns out to be unchanged, i.e., it is the Pr\"ahofer-Spohn function $\fKPZ(\cdot)$ [\figref{fig3}(d)] with correlation length growing as $\xi(t) \sim t^{2/3}$ [\figref{fig3}(e)].
Therefore, the KPZ physics remains intact in the presence of a finite energy current.
The propagation of the space-time correlation revealed in \figref{fig3}(a)-(c) is analogous to the case of growing tilted interfaces \cite{Ferrari-JSM2008,Corwin.etal-AIHPBPS2012} and nonlinear fluctuating hydrodynamics for unharmonic chains \cite{Mendl.Spohn-PRE2016}.
An important lesson from these studies is that one should measure the magnetization transfer $h(x,t)$ in the frame comoving with the space-time correlator, which amounts to the following expression:
\begin{equation}
    h(x,t) \equiv \int_0^t J(x,t')dt' - \int_{x-vt}^x m(x',0) dx'  \label{eq:height}
\end{equation}
with $v = v_\mathrm{peak}$.
With this appropriate definition of the magnetization transfer, we indeed confirm the KPZ growth of the variance, 
\eqref{eq:var} [\figref{fig3}(f) red squares], whereas the na\"ive definition $h(x,t) = \int_0^t J(x,t')dt'$ 
fails to capture the KPZ exponent (blue circles).
On the other hand, even with the definition \pref{eq:height} without left-right symmetry, we do not find any indication of asymmetric distribution, as evidenced by vanishing skewness [\figref{fig3}(g)].
The value of the kurtosis also remains far from that of the Baik-Rains distribution [\figref{fig3}(h)], just like the case without energy current \cite{Krajnik.etal-PRL2024,Rosenberg.etal-S2024}.
To summarize, the presence of a finite energy current only necessitates considering the comoving frame; otherwise, it seems to have no effect on relevant statistical quantities, as long as they are measured in the comoving frame.

\begin{figure}[tb]
\centering
\includegraphics[width=\hsize]{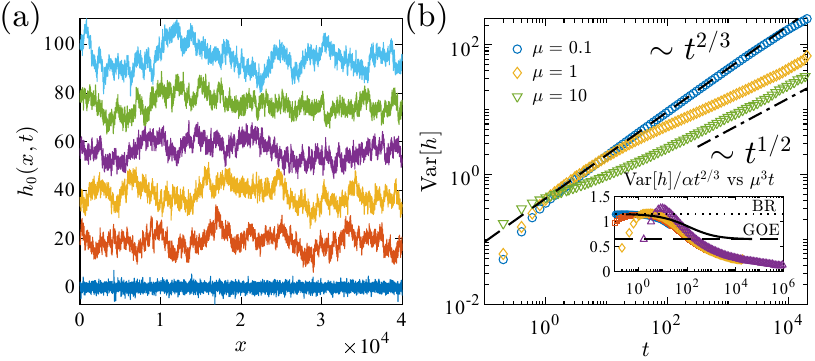}
\caption{
Results for KPLL with the flat initial condition.
(a) Snapshots of the height $h_0(x,t) = h(x,t) + h_0(x,0)$ at $t=0, 2000, 4000, \cdots, 10000$ from bottom to top, for $\mu=1$. For visibility, every subsequent snapshot is shifted upward by $20$.
(b) Variance of the magnetization transfer $h$ for different $\mu$.
The dashed line is the KPZ growth law \pref{eq:var} with $\alpha$ determined from the equilibrium simulations [\figref{fig1}(d)].
The dashed-dotted line is a guide for the eyes showing Var$[h] \sim t^{1/2}$. 
Inset: rescaled variance Var$[h]/\alpha t^{2/3}$ against $\mu^3 t$, for $\mu = 0.1, 0.5, 1, 2$, from top to bottom at the leftmost part of the datasets.
The bold solid line displays the behavior for KPZ interfaces \cite{Takeuchi-PRL2013} (with arbitrary horizontal shift), 
showing crossover from the Baik-Rains (BR) distribution (dotted line) to the characteristic distribution for flat interfaces, namely the GOE Tracy-Widom distribution (dashed line). 
Simulation parameters were $L = 40,000$ and $N=1,000$ for $\mu = 0.1, 0.5, 1$ and $L = 400,000$ and $N=33$ for $\mu = 2, 10$.
}
\label{fig4}
\end{figure}

Finally, we study the effect of the initial condition. 
Universal statistical properties of the authentic KPZ class are known to depend on the initial condition, the three representative cases being the domain wall (curved interface), flat, and stationary initial conditions \cite{Takeuchi-PA2018}.
It is important to assess whether KPZ scaling laws for non-stationary cases can describe spin chains under the corresponding, non-equilibrium settings.
For the domain wall initial condition, recent simulations suggested that KPZ may be observed only for finite times, being eventually replaced by the diffusive scaling \cite{Krajnik.etal-PRL2022}.
This is argued to result from the violation of the SU(2) symmetry, due to the chemical potential $\mu$ used to prepare each domain of biased spins \cite{Gopalakrishnan.etal-PNAS2019}.
Compared to this, the fate of the flat initial condition 
is not clear and has not been studied to our knowledge, even if some recent simulations of quantum spin chains hint that KPZ behavior is also visible when starting from non-stationary states \cite{PhysRevLett.132.130401}.

We realize a flat initial condition, by drawing each spin $\bm{S}_j(0)$ from infinite-temperature equilibrium distribution with a space-dependent vectorial chemical potential $\bm{\mu}_j$, $\rho(\bm{S}_j) = \frac{|\bm{\mu}_j|}{4\pi \sinh |\bm{\mu}_j|} e^{\bm{\mu}_j \cdot \bm{S}_j}$.
The chemical potential is determined as follows: (i) $\bm{\mu}_1 = 0$, (ii) $\bm{\mu}_{j \geq 2} = -\mu \bm{S}_{j-1}^\mathrm{tot}/|\bm{S}_{j-1}^\mathrm{tot}|$ with $\bm{S}_{j-1}^\mathrm{tot} \equiv \sum_{j'=1}^{j-1} \bm{S}_{j'}(0)$ and $\mu>0$ \footnote{After this initialization, we shift all spins slightly so that $\sum_j S_j^z(0)=0$ holds strictly.}. 
This amounts to generating an initial height profile $h_0(x,0) \equiv -\sum_{0 \leq j <x } S_j^z(0)$ that looks like a trajectory of an Ornstein-Uhlenbeck process [\figref{fig4}(a) bottom curve] instead of a Brownian trajectory for the equilibrium case $\mu=0$.
For KPZ interfaces, we demonstrate with TASEP that such initial conditions result asymptotically in the flat KPZ statistics \cite{Quastel.Remenik-TAMS2019}, through a dynamical crossover from the stationary statistics (the Baik-Rains distribution) to the flat one (the GOE Tracy-Widom distribution) \cite{Takeuchi-PRL2013} without changing the scaling Var$[h] \sim t^{2/3}$
(\supfigref{S-figS:TASEP} \cite{suppl}).
In contrast, for the KPLL magnet, we find completely different behavior for $\mu>0$, showing crossover from the KPZ scaling $t^{2/3}$ to the diffusive one $t^{1/2}$ 
[\figref{fig4}(b)].
Close scrutiny reveals that this crossover takes place at time scale $\mu^{-3}$ [\figref{fig4}(b) inset], in agreement with anomalous relaxation discussed in Ref.\,\cite{Gopalakrishnan.etal-PNAS2019}.
This indicates that the local violation of the isotropy (SU(2) for quantum spins) is sufficient for KPZ to break down in spin chains.

\begin{table}[tb]
 \caption{KPZ two-point properties in equilibrium integrable spin chains. The new results of this work are marked with *.}
 \label{tbl:comparison}
 \catcode`?=\active \def?{\phantom{0}}
 \begin{tabular}{lcc} \hline
 exponents & \multicolumn{2}{c}{$\xi(t)\sim t^{2/3}, ~\mathrm{Var}[h(x,t)] \sim t^{2/3}$} \\
 two-point function $C_2(\ell,t)$ & \multicolumn{2}{c}{Pr\"ahofer-Spohn solution $\fKPZ$} \\ 
 *variance amplitude & \multicolumn{2}{c}{\eqref{eq:var} with $\mathrm{Var[BR]} \approx 1.15$} \\
 *spatial correlation $C_s(\ell,t)$ & \multicolumn{2}{c}{Airy$_0$ covariance, \figref{fig2}(a)} \\
 *time correlation $C_t(t_1,t_2)$ & \multicolumn{2}{c}{Ferrari-Spohn solution, \eqref{eq:FS}} \\ \hline
 \end{tabular}%
\end{table}

In summary, using the integrable isotropic spin chains, both classical and quantum, we carried out quantitative tests of KPZ scaling laws for various two-point quantities that have not been characterized for spin chains so far, and found precise agreement in all of them (\tblref{tbl:comparison}).
Nevertheless, the KPZ scaling laws seem to not describe higher-order quantities, as evidenced by earlier studies \cite{Krajnik.etal-PRL2024,Rosenberg.etal-S2024}.
Therefore, as the main conclusion of the Letter, the strict KPZ class rules only a subset of statistical properties of isotropic integrable spin chains (and other cases with a continuous non-Abelian symmetry \cite{Ilievski.etal-PRX2021}).
It is of primary importance to clarify the underlying principles of such partial emergence of the KPZ class.
The coupled Burgers equations in Ref.\,\cite{DeNardis.etal-PRL2023} showed how KPZ two-point quantities can emerge out of symmetric one-point distributions, but the kurtosis remains unexplained. 
It is an open question if one can extend the theory to make it fully consistent, e.g., by introducing a larger number of hydrodynamic modes.
Our finding on the robustness of the KPZ scaling in the presence of energy current, as well as its breakdown by the local violation of isotropy, may also be hints for probing this mystery, which hangs over such simple quantum many-body systems as the isotropic Heisenberg spin chain.

\begin{acknowledgments}
\paragraph{Acknowledgments.} 
The authors acknowledge that this work deeply benefited from collaboration with Sarang Gopalakrishnan. K.A.T. also appreciates enlightening discussions with I. Bloch, H. Spohn, D. Wei, and J. Zeiher, as well as useful exchanges in Les Houches School of Physics ``Kardar–Parisi–Zhang equation: new trends in theories and experiments'' (2024). K.T. thanks T. Nakamoto and T. Soejima for the helpful discussion on tensor network simulation.  
The work is supported in part by KAKENHI from Japan Society for the Promotion of Science (Grant Nos. JP23K17664 (K.A.T. and K.T.), JP20H01826, JP19H05800, JP24K00593 (K.A.T.) and JP22K20350 (K.T.)), by JST PRESTO (Grant No.~JPMJPR2256) (K.T.), by the Deutsche Forschungsgemeinschaft (DFG, German Research Foundation) under Germany’s Excellence Strategy - GZ 2047/1, projekt-id 390685813 and by the Projektnummer 211504053 - SFB 1060 (P.L.F.), by ANR-22-CPJ1-0021-01 and ERC Starting Grant 101042293
(HEPIQ) (J.D.N), and by the US Department of Energy, Office of Science, Basic Energy Sciences, under award No. DE-SC0023999 (R.V.). 
\end{acknowledgments}

\section{End Matter}
\paragraph{Quantum model simulation.} 
To obtain the data shown in \figref{figQ}, we numerically evaluate $C_2(\ell, t)$, $C_s(\ell, t)$ and $C_t(\ell, t)$ for the infinite temperature equilibrium state. These quantities are defined for the quantum model as follows: $C_2(\ell, t)\equiv \langle\hat{S}_{0}^z \hat{S}_{\ell}^z(t)\rangle$, $C_s(\ell, t)\equiv \langle \Delta\hat{S}_{0} (t)\Delta\hat{S}_{\ell}(t)\rangle -  \langle \Delta\hat{S}_{\ell}(t)\rangle^2$, and $C_t(t_1, t_2)\equiv \langle \Delta\hat{S}_{0}(t_1) \Delta\hat{S}_{0}(t_2)\rangle- \langle \Delta \hat{S}_{0}(t_1)\rangle \langle \Delta \hat{S}_{0}(t_2)\rangle$ where $\langle \hat{\mathcal{O}} \rangle \equiv \mathrm{Tr}[\hat{\mathcal{O}}]/2^L$, $\hat{S}_j^\alpha(t)\equiv\hat{U}_t \hat{S}_j^\alpha \hat{U}_t^\dagger$~$(\alpha=x,y,z)$ with $\hat{U}_t \equiv e^{-i\hat{H}t}$, $\hat{J}_j(t) \equiv \hat{S}^x_j(t) \hat{S}^y_{j+1}(t) -\hat{S}^y_j(t) \hat{S}^x_{j+1}(t)$ and $\Delta\hat{S}_j(t) \equiv \int_0^t \hat{J}_j(t') dt'$ represents the local spin current and magnetization transfer at the $j$-th site respectively. Here, the site index $0$ represents the $(L/2)$-th site counted from the left boundary and is nearly at the center of the system, where $L$ is the total number of sites. The index $j$ runs from $-(L/2)+1$ to $L/2$.
In the quantum context, the definition of charge transfer must come with a prescription for the appropriate time-ordering of the operator $\Delta \hat{S}_j(t)$ as the current operators at different times do not commute with each other; see Refs~\cite{Levitov1993, Levitov1996}. However, for simplicity, our correlation functions $C_s(\ell, t)$ and $C_t(t_1, t_2)$ are not defined with such a time-ordering, even though these quantities are still physically observable in principle, as we argue in \suptxtref{3} \cite{suppl}. Testing the KPZ scaling for correlation functions with the appropriate time ordering to reproduce charge transfer is an important problem left for future studies.

To obtain the time-evolved operators and their expectation values, we use the time-evolving block decimation (TEBD) method \cite{Vidal-PRL2004} for matrix product operators, with time step $\Delta t=0.05$ and maximum bond dimension $\chi=400, 800, 1200, 1600$, implemented by ITensor \cite{ITensor}.
We consider $L=100$ and the open boundary condition for this simulation. 
For $C_s(\ell, t)$ and $C_t(t_1, t_2)$, we first obtain the current-current correlation function  $F_\ell(t) \equiv \langle \hat{J}_\ell(0) \hat{U}_t \hat{J}_0(0) \hat{U}_t^\dagger \rangle$ using the TEBD and then compute numerical integrals, such as $\langle \Delta\hat{S}_{0} (t)\Delta\hat{S}_{\ell}(t)\rangle = \int_0^t dt' \int_0^{t} dt'' F_\ell(t'-t'')$. The time integral is evaluated with a large time step $20\Delta t$. 

For simulations of such quantum models, it is crucial to deal with truncation error due to the finite bond dimension $\chi$. To understand the effect of this error, we have performed simulations with different $\chi$ and compared them. 
The results described in the following indicate that a large value of $\chi$ is needed for quantitative verification of the KPZ scaling functions, larger than that needed for the exponents.
While most of the data shown in \figref{figQ} are well-converged with respect to $\chi$ within numerical accuracy, $\alpha_1(t)$ is not, and we have systematically evaluated the error bars as explained below. Note that 
the Trotter error is not the dominant source of error in the present study.

\paragraph{Numerical results for the quantum model.}
Let us start with the quantities related to $C_2(\ell, t)$. Figure~\ref{figQ2}(a) shows $\Omega(t)\equiv \sum_l C_2(\ell, t)$, which 
should be conserved by the $U(1)$ symmetry of the Hamiltonian [\eqref{eq:Heisenberg}]. However, because of the truncation error due to finite bond dimension, $\Omega(t)$ starts to deviate at long times. Indeed, we find that the deviation becomes smaller with increasing $\chi$. 
With $\chi=1600$ used for the main results shown in \figref{figQ}, the deviation remained within $0.8\%$ of the initial value, at the largest time inspected in this work ($t=50$).
Next, we evaluate the correlation length $\xi(t)\equiv [\frac{1}{\sigma^2 \Omega(t)}\sum_\ell \ell^2 C_2(\ell, t)]^{1/2}$ with $\sigma^2 \approx 0.51$. The data exhibit the power law $\xi(t) \sim t^{2/3}$ [\figref{figQ2}(b)] as expected from the KPZ scaling.

Using these results, we plot the rescaled two-point function $\tilde{C}_2(\ell, t) = \frac{\xi(t)}{\Omega(t)}C_2(\ell,t)$ in \figref{figQ2}(c). The data are in good agreement with the Pr\"ahofer-Spohn exact solution $\fKPZ(u)$ with $u = \ell/\xi(t)$, confirming the earlier finding in Ref.\,\cite{Ljubotina.etal-PRL2019} without adjustable parameters.
However, to evaluate the parameter $\alpha_1(t)$ therefrom (\eqref{eq:C2} with $\alpha=\alpha_1(t)$), we must deal with the slight deviation from $\fKPZ(u)$ visible even for the largest $\chi$ we used.
To this end, in the inset of \figref{figQ2}(c),
we plot $\tilde{\alpha}_1(\ell, t) \equiv \xi(t)^2 C_2(\ell, t)/(2 t^{2/3} \fKPZ(\ell / \xi(t)))$. According to \eqref{eq:C2}, this quantity is expected to take a constant value $\alpha$ in the long-time limit.
Indeed, as $t$ increases, $\tilde{\alpha}_1(\ell, t)$ develops a plateau.
The deviation from the plateau gives a measure of the error in the estimate $\alpha_1(t)$.
Specifically, we measure $\alpha_{1}^\mathrm{max}(t) \equiv \max_u \tilde{\alpha}_1(u,t)$ and $\alpha_{1}^\mathrm{min}(t) \equiv \min_{|u| < 1} \tilde{\alpha}_1(u,t)$ [\figref{figQ2}(d)] and use $\frac{\alpha_1^\mathrm{max}(t) + \alpha_1^\mathrm{min}(t)}{2}$ as the estimate of $\alpha_1(t)$ and $\alpha_1^\mathrm{max}(t) - \alpha_1^\mathrm{min}(t)$ as its error bar.
The results in \figref{figQ2}(d) indeed show that $\alpha_1(t)$ tends to converge to a constant in the long-time and large-$\chi$ limit. 

\begin{figure}[tb]
  \centering
  \includegraphics[width=\hsize]{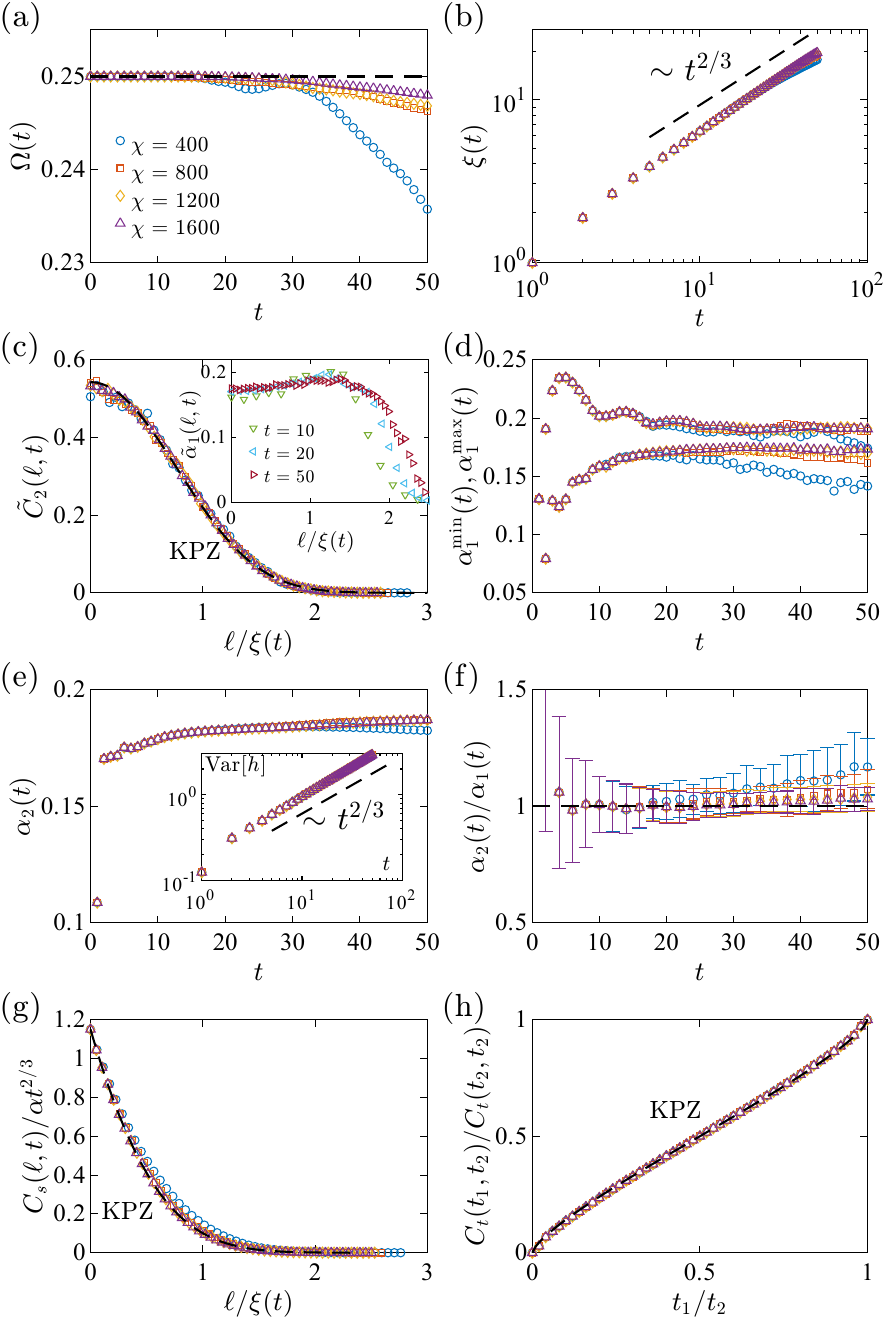}
  \caption{
  Detailed results for the quantum Heisenberg model (system size $L=100$ and open boundary condition). All panels except the inset of (c) show the $\chi$-dependence of $\Omega(t)$ (a), $\xi(t)$ (b), $\tilde{C}_2(\ell,t)$ (c), $\alpha_1^\mathrm{min}(t)$ and $\alpha_1^\mathrm{max}(t)$ (d), $\alpha_2(t)$ (e), $\mathrm{Var}[h]$ [(e) inset], $\alpha_2(t)/\alpha_1(t)$ (f), $C_s(\ell,t)$ (g), and $C_t(t_1,t_2)$ (h), where $\chi$ denotes the maximum bond dimension in the TEBD simulation. The same symbols are used in these panels as indicated in the legend of (a). The error bars in (f) are given by the range $[\alpha_2(t)/\alpha_1^\mathrm{max}(t), \alpha_2(t)/\alpha_1^\mathrm{min}(t)]$. The dashed lines in (c)(g)(h) are the Pr\"ahofer-Spohn solution $\fKPZ(u)$, the Airy$_0$ covariance $\mathcal{C}_0(u)$, and the Ferrari-Spohn solution \pref{eq:FS}, respectively. The inset of (c) shows the time dependence of $\tilde{\alpha}_1(\ell,t)$ with $\chi = 1600$.    
  }
  \label{figQ2}
\end{figure}

Next, we consider the quantities related to the magnetization transfer. 
First, its variance exhibits the KPZ scaling law, $\mathrm{Var}[h(x,t)] \sim t^{2/3}$ [\figref{figQ2}(e) inset].
From this, we evaluate $\alpha_2(t) \equiv \mathrm{Var}[h(0,t)]/ (t^{2/3}\mathrm{Var[BR]})$ in \figref{figQ2}(e), where the data converge to a constant value in the long-time limit. 
The KPZ scaling laws predict $\alpha_1(t) = \alpha_2(t) =$ a constant $\alpha$ [\eqsref{eq:C2} and \pref{eq:var}].
This is confirmed in \figref{figQ2}(f), where the ratio $\alpha_2(t)/\alpha_1(t)$ approaches one in the long-time and large-$\chi$ limit. 
Finally, we plot the spatial and temporal correlation functions, $C_s(\ell,t)$ and $C_t(t_1, t_2)$, respectively, in the rescaled units in \figsref{figQ2}(g)(h). They converge remarkably with increasing $\chi$ and show good agreement with the Airy$_0$ covariance $\mathcal{C}_0(u)$ and the Ferrari-Spohn exact solution \pref{eq:FS}, respectively, without adjustable parameters. These results provide strong evidence that the KPZ scaling laws for the two-point quantities are valid in the quantum model too. 


\bibliography{ref}

\begin{thebibliography}{63}%
\makeatletter
\providecommand \@ifxundefined [1]{%
 \@ifx{#1\undefined}
}%
\providecommand \@ifnum [1]{%
 \ifnum #1\expandafter \@firstoftwo
 \else \expandafter \@secondoftwo
 \fi
}%
\providecommand \@ifx [1]{%
 \ifx #1\expandafter \@firstoftwo
 \else \expandafter \@secondoftwo
 \fi
}%
\providecommand \natexlab [1]{#1}%
\providecommand \enquote  [1]{``#1''}%
\providecommand \bibnamefont  [1]{#1}%
\providecommand \bibfnamefont [1]{#1}%
\providecommand \citenamefont [1]{#1}%
\providecommand \href@noop [0]{\@secondoftwo}%
\providecommand \href [0]{\begingroup \@sanitize@url \@href}%
\providecommand \@href[1]{\@@startlink{#1}\@@href}%
\providecommand \@@href[1]{\endgroup#1\@@endlink}%
\providecommand \@sanitize@url [0]{\catcode `\\12\catcode `\$12\catcode
  `\&12\catcode `\#12\catcode `\^12\catcode `\_12\catcode `\%12\relax}%
\providecommand \@@startlink[1]{}%
\providecommand \@@endlink[0]{}%
\providecommand \url  [0]{\begingroup\@sanitize@url \@url }%
\providecommand \@url [1]{\endgroup\@href {#1}{\urlprefix }}%
\providecommand \urlprefix  [0]{URL }%
\providecommand \Eprint [0]{\href }%
\providecommand \doibase [0]{https://doi.org/}%
\providecommand \selectlanguage [0]{\@gobble}%
\providecommand \bibinfo  [0]{\@secondoftwo}%
\providecommand \bibfield  [0]{\@secondoftwo}%
\providecommand \translation [1]{[#1]}%
\providecommand \BibitemOpen [0]{}%
\providecommand \bibitemStop [0]{}%
\providecommand \bibitemNoStop [0]{.\EOS\space}%
\providecommand \EOS [0]{\spacefactor3000\relax}%
\providecommand \BibitemShut  [1]{\csname bibitem#1\endcsname}%
\let\auto@bib@innerbib\@empty
\bibitem [{\citenamefont {Castro-Alvaredo}\ \emph {et~al.}(2016)\citenamefont
  {Castro-Alvaredo}, \citenamefont {Doyon},\ and\ \citenamefont
  {Yoshimura}}]{CastroAlvaredo.etal-PRX2016}%
  \BibitemOpen
  \bibfield  {author} {\bibinfo {author} {\bibfnamefont {O.~A.}\ \bibnamefont
  {Castro-Alvaredo}}, \bibinfo {author} {\bibfnamefont {B.}~\bibnamefont
  {Doyon}},\ and\ \bibinfo {author} {\bibfnamefont {T.}~\bibnamefont
  {Yoshimura}},\ }\bibfield  {title} {\bibinfo {title} {Emergent hydrodynamics
  in integrable quantum systems out of equilibrium},\ }\href
  {https://doi.org/10.1103/PhysRevX.6.041065} {\bibfield  {journal} {\bibinfo
  {journal} {Phys. Rev. X}\ }\textbf {\bibinfo {volume} {6}},\ \bibinfo {pages}
  {041065} (\bibinfo {year} {2016})}\BibitemShut {NoStop}%
\bibitem [{\citenamefont {Bertini}\ \emph {et~al.}(2016)\citenamefont
  {Bertini}, \citenamefont {Collura}, \citenamefont {De~Nardis},\ and\
  \citenamefont {Fagotti}}]{Bertini.etal-PRL2016}%
  \BibitemOpen
  \bibfield  {author} {\bibinfo {author} {\bibfnamefont {B.}~\bibnamefont
  {Bertini}}, \bibinfo {author} {\bibfnamefont {M.}~\bibnamefont {Collura}},
  \bibinfo {author} {\bibfnamefont {J.}~\bibnamefont {De~Nardis}},\ and\
  \bibinfo {author} {\bibfnamefont {M.}~\bibnamefont {Fagotti}},\ }\bibfield
  {title} {\bibinfo {title} {Transport in out-of-equilibrium {XXZ} chains:
  Exact profiles of charges and currents},\ }\href
  {https://doi.org/10.1103/PhysRevLett.117.207201} {\bibfield  {journal}
  {\bibinfo  {journal} {Phys. Rev. Lett.}\ }\textbf {\bibinfo {volume} {117}},\
  \bibinfo {pages} {207201} (\bibinfo {year} {2016})}\BibitemShut {NoStop}%
\bibitem [{\citenamefont {Doyon}(2020)}]{Doyon-SPPLN2020}%
  \BibitemOpen
  \bibfield  {author} {\bibinfo {author} {\bibfnamefont {B.}~\bibnamefont
  {Doyon}},\ }\bibfield  {title} {\bibinfo {title} {Lecture notes on
  generalised hydrodynamics},\ }\href
  {https://doi.org/10.21468/scipostphyslectnotes.18} {\bibfield  {journal}
  {\bibinfo  {journal} {{SciPost} Phys. Lect. Notes}\ }\textbf {\bibinfo
  {volume} {18}},\ \bibinfo {pages} {1} (\bibinfo {year} {2020})}\BibitemShut
  {NoStop}%
\bibitem [{\citenamefont {Bulchandani}\ \emph {et~al.}(2021)\citenamefont
  {Bulchandani}, \citenamefont {Gopalakrishnan},\ and\ \citenamefont
  {Ilievski}}]{Bulchandani.etal-JSM2021}%
  \BibitemOpen
  \bibfield  {author} {\bibinfo {author} {\bibfnamefont {V.~B.}\ \bibnamefont
  {Bulchandani}}, \bibinfo {author} {\bibfnamefont {S.}~\bibnamefont
  {Gopalakrishnan}},\ and\ \bibinfo {author} {\bibfnamefont {E.}~\bibnamefont
  {Ilievski}},\ }\bibfield  {title} {\bibinfo {title} {Superdiffusion in spin
  chains},\ }\href {https://doi.org/10.1088/1742-5468/ac12c7} {\bibfield
  {journal} {\bibinfo  {journal} {J. Stat. Mech.}\ }\textbf {\bibinfo {volume}
  {2021}},\ \bibinfo {pages} {084001} (\bibinfo {year} {2021})}\BibitemShut
  {NoStop}%
\bibitem [{\citenamefont {Gopalakrishnan}\ and\ \citenamefont
  {Vasseur}(2024)}]{Gopalakrishnan.Vasseur-ARCMP2024}%
  \BibitemOpen
  \bibfield  {author} {\bibinfo {author} {\bibfnamefont {S.}~\bibnamefont
  {Gopalakrishnan}}\ and\ \bibinfo {author} {\bibfnamefont {R.}~\bibnamefont
  {Vasseur}},\ }\bibfield  {title} {\bibinfo {title} {Superdiffusion from
  nonabelian symmetries in nearly integrable systems},\ }\href
  {https://doi.org/10.1146/annurev-conmatphys-032922-110710} {\bibfield
  {journal} {\bibinfo  {journal} {Annu. Rev. Condens. Matter Phys.}\ }\textbf
  {\bibinfo {volume} {15}},\ \bibinfo {pages} {159} (\bibinfo {year}
  {2024})}\BibitemShut {NoStop}%
\bibitem [{\citenamefont {Fabricius}\ and\ \citenamefont
  {McCoy}(1998)}]{Fabricius.McCoy-PRB1998}%
  \BibitemOpen
  \bibfield  {author} {\bibinfo {author} {\bibfnamefont {K.}~\bibnamefont
  {Fabricius}}\ and\ \bibinfo {author} {\bibfnamefont {B.~M.}\ \bibnamefont
  {McCoy}},\ }\bibfield  {title} {\bibinfo {title} {Spin diffusion and the
  spin-1/2 $\mathrm{XXZ}$ chain at $t=\ensuremath{\infty}$ from exact
  diagonalization},\ }\href {https://doi.org/10.1103/PhysRevB.57.8340}
  {\bibfield  {journal} {\bibinfo  {journal} {Phys. Rev. B}\ }\textbf {\bibinfo
  {volume} {57}},\ \bibinfo {pages} {8340} (\bibinfo {year}
  {1998})}\BibitemShut {NoStop}%
\bibitem [{\citenamefont {Ljubotina}\ \emph {et~al.}(2017)\citenamefont
  {Ljubotina}, \citenamefont {{\v{Z}}nidari{\v{c}}},\ and\ \citenamefont
  {Prosen}}]{Ljubotina.etal-NC2017}%
  \BibitemOpen
  \bibfield  {author} {\bibinfo {author} {\bibfnamefont {M.}~\bibnamefont
  {Ljubotina}}, \bibinfo {author} {\bibfnamefont {M.}~\bibnamefont
  {{\v{Z}}nidari{\v{c}}}},\ and\ \bibinfo {author} {\bibfnamefont
  {T.}~\bibnamefont {Prosen}},\ }\bibfield  {title} {\bibinfo {title} {Spin
  diffusion from an inhomogeneous quench in an integrable system},\ }\href
  {https://doi.org/10.1038/ncomms16117} {\bibfield  {journal} {\bibinfo
  {journal} {Nat. Commun.}\ }\textbf {\bibinfo {volume} {8}},\ \bibinfo {pages}
  {16117} (\bibinfo {year} {2017})}\BibitemShut {NoStop}%
\bibitem [{\citenamefont {Gopalakrishnan}\ and\ \citenamefont
  {Vasseur}(2019)}]{GV2019}%
  \BibitemOpen
  \bibfield  {author} {\bibinfo {author} {\bibfnamefont {S.}~\bibnamefont
  {Gopalakrishnan}}\ and\ \bibinfo {author} {\bibfnamefont {R.}~\bibnamefont
  {Vasseur}},\ }\bibfield  {title} {\bibinfo {title} {Kinetic theory of spin
  diffusion and superdiffusion in {XXZ} spin chains},\ }\href
  {https://doi.org/10.1103/PhysRevLett.122.127202} {\bibfield  {journal}
  {\bibinfo  {journal} {Phys. Rev. Lett.}\ }\textbf {\bibinfo {volume} {122}},\
  \bibinfo {pages} {127202} (\bibinfo {year} {2019})}\BibitemShut {NoStop}%
\bibitem [{\citenamefont {De~Nardis}\ \emph {et~al.}(2019)\citenamefont
  {De~Nardis}, \citenamefont {Medenjak}, \citenamefont {Karrasch},\ and\
  \citenamefont {Ilievski}}]{PhysRevLett.123.186601}%
  \BibitemOpen
  \bibfield  {author} {\bibinfo {author} {\bibfnamefont {J.}~\bibnamefont
  {De~Nardis}}, \bibinfo {author} {\bibfnamefont {M.}~\bibnamefont {Medenjak}},
  \bibinfo {author} {\bibfnamefont {C.}~\bibnamefont {Karrasch}},\ and\
  \bibinfo {author} {\bibfnamefont {E.}~\bibnamefont {Ilievski}},\ }\bibfield
  {title} {\bibinfo {title} {Anomalous spin diffusion in one-dimensional
  antiferromagnets},\ }\href {https://doi.org/10.1103/PhysRevLett.123.186601}
  {\bibfield  {journal} {\bibinfo  {journal} {Phys. Rev. Lett.}\ }\textbf
  {\bibinfo {volume} {123}},\ \bibinfo {pages} {186601} (\bibinfo {year}
  {2019})}\BibitemShut {NoStop}%
\bibitem [{\citenamefont {Takeuchi}(2018)}]{Takeuchi-PA2018}%
  \BibitemOpen
  \bibfield  {author} {\bibinfo {author} {\bibfnamefont {K.~A.}\ \bibnamefont
  {Takeuchi}},\ }\bibfield  {title} {\bibinfo {title} {An appetizer to modern
  developments on the {Kardar-Parisi-Zhang} universality class},\ }\href
  {https://doi.org/10.1016/j.physa.2018.03.009} {\bibfield  {journal} {\bibinfo
   {journal} {Physica A}\ }\textbf {\bibinfo {volume} {504}},\ \bibinfo {pages}
  {77} (\bibinfo {year} {2018})}\BibitemShut {NoStop}%
\bibitem [{\citenamefont {Corwin}(2012)}]{Corwin-RMTA2012}%
  \BibitemOpen
  \bibfield  {author} {\bibinfo {author} {\bibfnamefont {I.}~\bibnamefont
  {Corwin}},\ }\bibfield  {title} {\bibinfo {title} {The {Kardar-Parisi-Zhang}
  equation and universality class},\ }\href
  {https://doi.org/10.1142/S2010326311300014} {\bibfield  {journal} {\bibinfo
  {journal} {Random Matrices Theory Appl.}\ }\textbf {\bibinfo {volume} {1}},\
  \bibinfo {pages} {1130001} (\bibinfo {year} {2012})}\BibitemShut {NoStop}%
\bibitem [{\citenamefont {Prolhac}(2024)}]{Prolhac-a2024}%
  \BibitemOpen
  \bibfield  {author} {\bibinfo {author} {\bibfnamefont {S.}~\bibnamefont
  {Prolhac}},\ }\bibfield  {title} {\bibinfo {title} {{KPZ} fluctuations in
  finite volume},\ }\href {https://doi.org/10.21468/scipostphyslectnotes.81}
  {\bibfield  {journal} {\bibinfo  {journal} {SciPost Phys. Lect. Notes}\ ,\
  \bibinfo {pages} {81}} (\bibinfo {year} {2024})}\BibitemShut {NoStop}%
\bibitem [{\citenamefont {Ljubotina}\ \emph {et~al.}(2019)\citenamefont
  {Ljubotina}, \citenamefont {\v{Z}nidari\v{c}},\ and\ \citenamefont
  {Prosen}}]{Ljubotina.etal-PRL2019}%
  \BibitemOpen
  \bibfield  {author} {\bibinfo {author} {\bibfnamefont {M.}~\bibnamefont
  {Ljubotina}}, \bibinfo {author} {\bibfnamefont {M.}~\bibnamefont
  {\v{Z}nidari\v{c}}},\ and\ \bibinfo {author} {\bibfnamefont {T.}~\bibnamefont
  {Prosen}},\ }\bibfield  {title} {\bibinfo {title} {{Kardar-Parisi-Zhang}
  physics in the quantum {Heisenberg} magnet},\ }\href
  {https://doi.org/10.1103/PhysRevLett.122.210602} {\bibfield  {journal}
  {\bibinfo  {journal} {Phys. Rev. Lett.}\ }\textbf {\bibinfo {volume} {122}},\
  \bibinfo {pages} {210602} (\bibinfo {year} {2019})}\BibitemShut {NoStop}%
\bibitem [{\citenamefont {Weiner}\ \emph {et~al.}(2020)\citenamefont {Weiner},
  \citenamefont {Schmitteckert}, \citenamefont {Bera},\ and\ \citenamefont
  {Evers}}]{Weiner.etal-PRB2020}%
  \BibitemOpen
  \bibfield  {author} {\bibinfo {author} {\bibfnamefont {F.}~\bibnamefont
  {Weiner}}, \bibinfo {author} {\bibfnamefont {P.}~\bibnamefont
  {Schmitteckert}}, \bibinfo {author} {\bibfnamefont {S.}~\bibnamefont
  {Bera}},\ and\ \bibinfo {author} {\bibfnamefont {F.}~\bibnamefont {Evers}},\
  }\bibfield  {title} {\bibinfo {title} {High-temperature spin dynamics in the
  heisenberg chain: Magnon propagation and emerging {Kardar-Parisi-Zhang}
  scaling in the zero-magnetization limit},\ }\href
  {https://doi.org/10.1103/PhysRevB.101.045115} {\bibfield  {journal} {\bibinfo
   {journal} {Phys. Rev. B}\ }\textbf {\bibinfo {volume} {101}},\ \bibinfo
  {pages} {045115} (\bibinfo {year} {2020})}\BibitemShut {NoStop}%
\bibitem [{\citenamefont {Pr\"ahofer}\ and\ \citenamefont
  {Spohn}(2004)}]{Prahofer.Spohn-JSP2004}%
  \BibitemOpen
  \bibfield  {author} {\bibinfo {author} {\bibfnamefont {M.}~\bibnamefont
  {Pr\"ahofer}}\ and\ \bibinfo {author} {\bibfnamefont {H.}~\bibnamefont
  {Spohn}},\ }\bibfield  {title} {\bibinfo {title} {Exact scaling functions for
  one-dimensional stationary {KPZ} growth},\ }\href
  {http://dx.doi.org/10.1023/B:JOSS.0000019810.21828.fc} {\bibfield  {journal}
  {\bibinfo  {journal} {J. Stat. Phys.}\ }\textbf {\bibinfo {volume} {115}},\
  \bibinfo {pages} {255} (\bibinfo {year} {2004})}\BibitemShut {NoStop}%
\bibitem [{\citenamefont {Ilievski}\ \emph {et~al.}(2021)\citenamefont
  {Ilievski}, \citenamefont {De~Nardis}, \citenamefont {Gopalakrishnan},
  \citenamefont {Vasseur},\ and\ \citenamefont {Ware}}]{Ilievski.etal-PRX2021}%
  \BibitemOpen
  \bibfield  {author} {\bibinfo {author} {\bibfnamefont {E.}~\bibnamefont
  {Ilievski}}, \bibinfo {author} {\bibfnamefont {J.}~\bibnamefont {De~Nardis}},
  \bibinfo {author} {\bibfnamefont {S.}~\bibnamefont {Gopalakrishnan}},
  \bibinfo {author} {\bibfnamefont {R.}~\bibnamefont {Vasseur}},\ and\ \bibinfo
  {author} {\bibfnamefont {B.}~\bibnamefont {Ware}},\ }\bibfield  {title}
  {\bibinfo {title} {Superuniversality of superdiffusion},\ }\href
  {https://doi.org/10.1103/PhysRevX.11.031023} {\bibfield  {journal} {\bibinfo
  {journal} {Phys. Rev. X}\ }\textbf {\bibinfo {volume} {11}},\ \bibinfo
  {pages} {031023} (\bibinfo {year} {2021})}\BibitemShut {NoStop}%
\bibitem [{\citenamefont {Ye}\ \emph {et~al.}(2022)\citenamefont {Ye},
  \citenamefont {Machado}, \citenamefont {Kemp}, \citenamefont {Hutson},\ and\
  \citenamefont {Yao}}]{Ye.etal-PRL2022}%
  \BibitemOpen
  \bibfield  {author} {\bibinfo {author} {\bibfnamefont {B.}~\bibnamefont
  {Ye}}, \bibinfo {author} {\bibfnamefont {F.}~\bibnamefont {Machado}},
  \bibinfo {author} {\bibfnamefont {J.}~\bibnamefont {Kemp}}, \bibinfo {author}
  {\bibfnamefont {R.~B.}\ \bibnamefont {Hutson}},\ and\ \bibinfo {author}
  {\bibfnamefont {N.~Y.}\ \bibnamefont {Yao}},\ }\bibfield  {title} {\bibinfo
  {title} {Universal {Kardar-Parisi-Zhang} dynamics in integrable quantum
  systems},\ }\href {https://doi.org/10.1103/PhysRevLett.129.230602} {\bibfield
   {journal} {\bibinfo  {journal} {Phys. Rev. Lett.}\ }\textbf {\bibinfo
  {volume} {129}},\ \bibinfo {pages} {230602} (\bibinfo {year}
  {2022})}\BibitemShut {NoStop}%
\bibitem [{\citenamefont {Das}\ \emph {et~al.}(2019)\citenamefont {Das},
  \citenamefont {Kulkarni}, \citenamefont {Spohn},\ and\ \citenamefont
  {Dhar}}]{Das.etal-PRE2019}%
  \BibitemOpen
  \bibfield  {author} {\bibinfo {author} {\bibfnamefont {A.}~\bibnamefont
  {Das}}, \bibinfo {author} {\bibfnamefont {M.}~\bibnamefont {Kulkarni}},
  \bibinfo {author} {\bibfnamefont {H.}~\bibnamefont {Spohn}},\ and\ \bibinfo
  {author} {\bibfnamefont {A.}~\bibnamefont {Dhar}},\ }\bibfield  {title}
  {\bibinfo {title} {{Kardar-Parisi-Zhang} scaling for an integrable lattice
  {Landau-Lifshitz} spin chain},\ }\href
  {https://doi.org/10.1103/PhysRevE.100.042116} {\bibfield  {journal} {\bibinfo
   {journal} {Phys. Rev. E}\ }\textbf {\bibinfo {volume} {100}},\ \bibinfo
  {pages} {042116} (\bibinfo {year} {2019})}\BibitemShut {NoStop}%
\bibitem [{\citenamefont {Scheie}\ \emph {et~al.}(2021)\citenamefont {Scheie},
  \citenamefont {Sherman}, \citenamefont {Dupont}, \citenamefont {Nagler},
  \citenamefont {Stone}, \citenamefont {Granroth}, \citenamefont {Moore},\ and\
  \citenamefont {Tennant}}]{Scheie.etal-NP2021}%
  \BibitemOpen
  \bibfield  {author} {\bibinfo {author} {\bibfnamefont {A.}~\bibnamefont
  {Scheie}}, \bibinfo {author} {\bibfnamefont {N.~E.}\ \bibnamefont {Sherman}},
  \bibinfo {author} {\bibfnamefont {M.}~\bibnamefont {Dupont}}, \bibinfo
  {author} {\bibfnamefont {S.~E.}\ \bibnamefont {Nagler}}, \bibinfo {author}
  {\bibfnamefont {M.~B.}\ \bibnamefont {Stone}}, \bibinfo {author}
  {\bibfnamefont {G.~E.}\ \bibnamefont {Granroth}}, \bibinfo {author}
  {\bibfnamefont {J.~E.}\ \bibnamefont {Moore}},\ and\ \bibinfo {author}
  {\bibfnamefont {D.~A.}\ \bibnamefont {Tennant}},\ }\bibfield  {title}
  {\bibinfo {title} {Detection of {Kardar-Parisi-Zhang} hydrodynamics in a
  quantum heisenberg spin-1/2 chain},\ }\href
  {https://doi.org/10.1038/s41567-021-01191-6} {\bibfield  {journal} {\bibinfo
  {journal} {Nat. Phys.}\ }\textbf {\bibinfo {volume} {17}},\ \bibinfo {pages}
  {726} (\bibinfo {year} {2021})}\BibitemShut {NoStop}%
\bibitem [{\citenamefont {Wei}\ \emph {et~al.}(2022)\citenamefont {Wei},
  \citenamefont {Rubio-Abadal}, \citenamefont {Ye}, \citenamefont {Machado},
  \citenamefont {Kemp}, \citenamefont {Srakaew}, \citenamefont {Hollerith},
  \citenamefont {Rui}, \citenamefont {Gopalakrishnan}, \citenamefont {Yao},
  \citenamefont {Bloch},\ and\ \citenamefont {Zeiher}}]{Wei.etal-S2022}%
  \BibitemOpen
  \bibfield  {author} {\bibinfo {author} {\bibfnamefont {D.}~\bibnamefont
  {Wei}}, \bibinfo {author} {\bibfnamefont {A.}~\bibnamefont {Rubio-Abadal}},
  \bibinfo {author} {\bibfnamefont {B.}~\bibnamefont {Ye}}, \bibinfo {author}
  {\bibfnamefont {F.}~\bibnamefont {Machado}}, \bibinfo {author} {\bibfnamefont
  {J.}~\bibnamefont {Kemp}}, \bibinfo {author} {\bibfnamefont {K.}~\bibnamefont
  {Srakaew}}, \bibinfo {author} {\bibfnamefont {S.}~\bibnamefont {Hollerith}},
  \bibinfo {author} {\bibfnamefont {J.}~\bibnamefont {Rui}}, \bibinfo {author}
  {\bibfnamefont {S.}~\bibnamefont {Gopalakrishnan}}, \bibinfo {author}
  {\bibfnamefont {N.~Y.}\ \bibnamefont {Yao}}, \bibinfo {author} {\bibfnamefont
  {I.}~\bibnamefont {Bloch}},\ and\ \bibinfo {author} {\bibfnamefont
  {J.}~\bibnamefont {Zeiher}},\ }\bibfield  {title} {\bibinfo {title} {Quantum
  gas microscopy of {Kardar-Parisi-Zhang} superdiffusion},\ }\href
  {https://doi.org/10.1126/science.abk2397} {\bibfield  {journal} {\bibinfo
  {journal} {Science}\ }\textbf {\bibinfo {volume} {376}},\ \bibinfo {pages}
  {716} (\bibinfo {year} {2022})}\BibitemShut {NoStop}%
\bibitem [{\citenamefont {Krajnik}\ \emph {et~al.}(2022)\citenamefont
  {Krajnik}, \citenamefont {Ilievski},\ and\ \citenamefont
  {Prosen}}]{Krajnik.etal-PRL2022}%
  \BibitemOpen
  \bibfield  {author} {\bibinfo {author} {\bibfnamefont {{\v{Z}}.}~\bibnamefont
  {Krajnik}}, \bibinfo {author} {\bibfnamefont {E.}~\bibnamefont {Ilievski}},\
  and\ \bibinfo {author} {\bibfnamefont {T.}~\bibnamefont {Prosen}},\
  }\bibfield  {title} {\bibinfo {title} {Absence of normal fluctuations in an
  integrable magnet},\ }\href {https://doi.org/10.1103/PhysRevLett.128.090604}
  {\bibfield  {journal} {\bibinfo  {journal} {Phys. Rev. Lett.}\ }\textbf
  {\bibinfo {volume} {128}},\ \bibinfo {pages} {090604} (\bibinfo {year}
  {2022})}\BibitemShut {NoStop}%
\bibitem [{\citenamefont {Krajnik}\ \emph {et~al.}(2024)\citenamefont
  {Krajnik}, \citenamefont {Schmidt}, \citenamefont {Ilievski},\ and\
  \citenamefont {Prosen}}]{Krajnik.etal-PRL2024}%
  \BibitemOpen
  \bibfield  {author} {\bibinfo {author} {\bibfnamefont {{\v{Z}}.}~\bibnamefont
  {Krajnik}}, \bibinfo {author} {\bibfnamefont {J.}~\bibnamefont {Schmidt}},
  \bibinfo {author} {\bibfnamefont {E.}~\bibnamefont {Ilievski}},\ and\
  \bibinfo {author} {\bibfnamefont {T.}~\bibnamefont {Prosen}},\ }\bibfield
  {title} {\bibinfo {title} {Dynamical criticality of magnetization transfer in
  integrable spin chains},\ }\href
  {https://doi.org/10.1103/PhysRevLett.132.017101} {\bibfield  {journal}
  {\bibinfo  {journal} {Phys. Rev. Lett.}\ }\textbf {\bibinfo {volume} {132}},\
  \bibinfo {pages} {017101} (\bibinfo {year} {2024})}\BibitemShut {NoStop}%
\bibitem [{\citenamefont {Rosenberg}\ \emph {et~al.}(2024)\citenamefont
  {Rosenberg} \emph {et~al.}}]{Rosenberg.etal-S2024}%
  \BibitemOpen
  \bibfield  {author} {\bibinfo {author} {\bibfnamefont {E.}~\bibnamefont
  {Rosenberg}} \emph {et~al.},\ }\bibfield  {title} {\bibinfo {title} {Dynamics
  of magnetization at infinite temperature in a {Heisenberg} spin chain},\
  }\href {https://doi.org/10.1126/science.adi7877} {\bibfield  {journal}
  {\bibinfo  {journal} {Science}\ }\textbf {\bibinfo {volume} {384}},\ \bibinfo
  {pages} {48} (\bibinfo {year} {2024})}\BibitemShut {NoStop}%
\bibitem [{\citenamefont {Quastel}\ and\ \citenamefont
  {Remenik}(2014)}]{Quastel.Remenik-Inbook2014}%
  \BibitemOpen
  \bibfield  {author} {\bibinfo {author} {\bibfnamefont {J.}~\bibnamefont
  {Quastel}}\ and\ \bibinfo {author} {\bibfnamefont {D.}~\bibnamefont
  {Remenik}},\ }\bibfield  {title} {\bibinfo {title} {Airy processes and
  variational problems},\ }in\ \href
  {https://doi.org/10.1007/978-1-4939-0339-9_5} {\emph {\bibinfo {booktitle}
  {Topics in Percolative and Disordered Systems}}},\ \bibinfo {series}
  {Springer Proceedings in Mathematics \& Statistics}, Vol.~\bibinfo {volume}
  {69},\ \bibinfo {editor} {edited by\ \bibinfo {editor} {\bibfnamefont
  {A.}~\bibnamefont {Ram\'irez}}, \bibinfo {editor} {\bibfnamefont
  {G.}~\bibnamefont {Ben~Arous}}, \bibinfo {editor} {\bibfnamefont
  {P.}~\bibnamefont {Ferrari}}, \bibinfo {editor} {\bibfnamefont
  {C.}~\bibnamefont {Newman}}, \bibinfo {editor} {\bibfnamefont
  {V.}~\bibnamefont {Sidoravicius}},\ and\ \bibinfo {editor} {\bibfnamefont
  {M.}~\bibnamefont {Vares}}}\ (\bibinfo  {publisher} {Springer},\ \bibinfo
  {address} {New York},\ \bibinfo {year} {2014})\ pp.\ \bibinfo {pages}
  {121--171},\ \Eprint {https://arxiv.org/abs/1301.0750} {arXiv:1301.0750}
  \BibitemShut {NoStop}%
\bibitem [{\citenamefont {Ferrari}\ and\ \citenamefont
  {Spohn}(2016)}]{Ferrari.Spohn-SIG2016}%
  \BibitemOpen
  \bibfield  {author} {\bibinfo {author} {\bibfnamefont {P.~L.}\ \bibnamefont
  {Ferrari}}\ and\ \bibinfo {author} {\bibfnamefont {H.}~\bibnamefont
  {Spohn}},\ }\bibfield  {title} {\bibinfo {title} {On time correlations for
  {KPZ} growth in one dimension},\ }\href
  {https://doi.org/10.3842/SIGMA.2016.074} {\bibfield  {journal} {\bibinfo
  {journal} {SIGMA}\ }\textbf {\bibinfo {volume} {12}},\ \bibinfo {pages} {074}
  (\bibinfo {year} {2016})}\BibitemShut {NoStop}%
\bibitem [{\citenamefont {Ferrari}\ and\ \citenamefont
  {Occelli}(2019)}]{Ferrari.Occelli-MPAG2019}%
  \BibitemOpen
  \bibfield  {author} {\bibinfo {author} {\bibfnamefont {P.~L.}\ \bibnamefont
  {Ferrari}}\ and\ \bibinfo {author} {\bibfnamefont {A.}~\bibnamefont
  {Occelli}},\ }\bibfield  {title} {\bibinfo {title} {Time-time covariance for
  last passage percolation with generic initial profile},\ }\href
  {https://doi.org/10.1007/s11040-018-9300-6} {\bibfield  {journal} {\bibinfo
  {journal} {Math. Phys. Anal. Geom.}\ }\textbf {\bibinfo {volume} {22}},\
  \bibinfo {pages} {1} (\bibinfo {year} {2019})}\BibitemShut {NoStop}%
\bibitem [{\citenamefont {Krajnik}\ and\ \citenamefont
  {Prosen}(2020)}]{Krajnik.Prosen-JSP2020}%
  \BibitemOpen
  \bibfield  {author} {\bibinfo {author} {\bibfnamefont {{\v{Z}}.}~\bibnamefont
  {Krajnik}}\ and\ \bibinfo {author} {\bibfnamefont {T.}~\bibnamefont
  {Prosen}},\ }\bibfield  {title} {\bibinfo {title} {{Kardar-Parisi-Zhang}
  physics in integrable rotationally symmetric dynamics on discrete space-time
  lattice},\ }\href {https://doi.org/10.1007/s10955-020-02523-1} {\bibfield
  {journal} {\bibinfo  {journal} {J. Stat. Phys.}\ }\textbf {\bibinfo {volume}
  {179}},\ \bibinfo {pages} {110} (\bibinfo {year} {2020})}\BibitemShut
  {NoStop}%
\bibitem [{\citenamefont {Krajnik}\ \emph {et~al.}(2021)\citenamefont
  {Krajnik}, \citenamefont {Ilievski}, \citenamefont {Prosen},\ and\
  \citenamefont {Pasquier}}]{Krajnik.etal-SPP2021}%
  \BibitemOpen
  \bibfield  {author} {\bibinfo {author} {\bibfnamefont {{\v{Z}}.}~\bibnamefont
  {Krajnik}}, \bibinfo {author} {\bibfnamefont {E.}~\bibnamefont {Ilievski}},
  \bibinfo {author} {\bibfnamefont {T.}~\bibnamefont {Prosen}},\ and\ \bibinfo
  {author} {\bibfnamefont {V.}~\bibnamefont {Pasquier}},\ }\bibfield  {title}
  {\bibinfo {title} {Anisotropic {Landau-Lifshitz} model in discrete
  space-time},\ }\href {https://doi.org/10.21468/scipostphys.11.3.051}
  {\bibfield  {journal} {\bibinfo  {journal} {SciPost Phys.}\ }\textbf
  {\bibinfo {volume} {11}},\ \bibinfo {pages} {051} (\bibinfo {year}
  {2021})}\BibitemShut {NoStop}%
\bibitem [{\citenamefont {Lakshmanan}(2011)}]{Lakshmanan-PTRSA2011}%
  \BibitemOpen
  \bibfield  {author} {\bibinfo {author} {\bibfnamefont {M.}~\bibnamefont
  {Lakshmanan}},\ }\bibfield  {title} {\bibinfo {title} {The fascinating world
  of the landau{\textendash}lifshitz{\textendash}gilbert equation: an
  overview},\ }\href {https://doi.org/10.1098/rsta.2010.0319} {\bibfield
  {journal} {\bibinfo  {journal} {Phil. Trans. R. Soc. A}\ }\textbf {\bibinfo
  {volume} {369}},\ \bibinfo {pages} {1280} (\bibinfo {year}
  {2011})}\BibitemShut {NoStop}%
\bibitem [{sup()}]{suppl}%
  \BibitemOpen
  \href@noop {} {}\bibinfo {note} {See Supplemental Material for Supplemental
  Texts~1-3, which include Refs.\cite{DOV22,EJS20,CP15b,FO17,TW94,QR12,FS05a},
  and for \supfigsref{S-figS:LL} and \ref{S-figS:TASEP}.}\BibitemShut {Stop}%
\bibitem [{\citenamefont {Ishimori}(1982)}]{Ishimori-JPSJ1982}%
  \BibitemOpen
  \bibfield  {author} {\bibinfo {author} {\bibfnamefont {Y.}~\bibnamefont
  {Ishimori}},\ }\bibfield  {title} {\bibinfo {title} {An integrable classical
  spin chain},\ }\href {https://doi.org/10.1143/jpsj.51.3417} {\bibfield
  {journal} {\bibinfo  {journal} {J. Phys. Soc. Jpn.}\ }\textbf {\bibinfo
  {volume} {51}},\ \bibinfo {pages} {3417} (\bibinfo {year}
  {1982})}\BibitemShut {NoStop}%
\bibitem [{\citenamefont {Takahashi}(1999)}]{Takahashi_book}%
  \BibitemOpen
  \bibfield  {author} {\bibinfo {author} {\bibfnamefont {M.}~\bibnamefont
  {Takahashi}},\ }\href@noop {} {\emph {\bibinfo {title} {Thermodynamics of
  One-Dimensional Solvable Models}}}\ (\bibinfo  {publisher} {Cambridge
  University Press},\ \bibinfo {year} {1999})\BibitemShut {NoStop}%
\bibitem [{\citenamefont {Baik}\ and\ \citenamefont
  {Rains}(2000)}]{Baik.Rains-JSP2000}%
  \BibitemOpen
  \bibfield  {author} {\bibinfo {author} {\bibfnamefont {J.}~\bibnamefont
  {Baik}}\ and\ \bibinfo {author} {\bibfnamefont {E.~M.}\ \bibnamefont
  {Rains}},\ }\bibfield  {title} {\bibinfo {title} {Limiting distributions for
  a polynuclear growth model with external sources},\ }\href
  {http://dx.doi.org/10.1023/A:1018615306992} {\bibfield  {journal} {\bibinfo
  {journal} {J. Stat. Phys.}\ }\textbf {\bibinfo {volume} {100}},\ \bibinfo
  {pages} {523} (\bibinfo {year} {2000})}\BibitemShut {NoStop}%
\bibitem [{\citenamefont {Iwatsuka}\ \emph {et~al.}(2020)\citenamefont
  {Iwatsuka}, \citenamefont {Fukai},\ and\ \citenamefont
  {Takeuchi}}]{Iwatsuka.etal-PRL2020}%
  \BibitemOpen
  \bibfield  {author} {\bibinfo {author} {\bibfnamefont {T.}~\bibnamefont
  {Iwatsuka}}, \bibinfo {author} {\bibfnamefont {Y.~T.}\ \bibnamefont
  {Fukai}},\ and\ \bibinfo {author} {\bibfnamefont {K.~A.}\ \bibnamefont
  {Takeuchi}},\ }\bibfield  {title} {\bibinfo {title} {Direct evidence for
  universal statistics of stationary {Kardar-Parisi-Zhang} interfaces},\ }\href
  {https://doi.org/10.1103/PhysRevLett.124.250602} {\bibfield  {journal}
  {\bibinfo  {journal} {Phys. Rev. Lett.}\ }\textbf {\bibinfo {volume} {124}},\
  \bibinfo {pages} {250602} (\bibinfo {year} {2020})}\BibitemShut {NoStop}%
\bibitem [{\citenamefont {Pr\"ahofer}\ and\ \citenamefont
  {Spohn}(2002)}]{Prahofer.Spohn-JSP2002}%
  \BibitemOpen
  \bibfield  {author} {\bibinfo {author} {\bibfnamefont {M.}~\bibnamefont
  {Pr\"ahofer}}\ and\ \bibinfo {author} {\bibfnamefont {H.}~\bibnamefont
  {Spohn}},\ }\bibfield  {title} {\bibinfo {title} {{Scale Invariance of the
  PNG Droplet and the Airy Process}},\ }\href
  {http://dx.doi.org/10.1023/A:1019791415147} {\bibfield  {journal} {\bibinfo
  {journal} {J. Stat. Phys.}\ }\textbf {\bibinfo {volume} {108}},\ \bibinfo
  {pages} {1071} (\bibinfo {year} {2002})}\BibitemShut {NoStop}%
\bibitem [{\citenamefont {Sasamoto}(2005)}]{Sasamoto-JPA2005}%
  \BibitemOpen
  \bibfield  {author} {\bibinfo {author} {\bibfnamefont {T.}~\bibnamefont
  {Sasamoto}},\ }\bibfield  {title} {\bibinfo {title} {{Spatial correlations of
  the 1D KPZ surface on a flat substrate}},\ }\href
  {http://stacks.iop.org/0305-4470/38/i=33/a=L01} {\bibfield  {journal}
  {\bibinfo  {journal} {J. Phys. A}\ }\textbf {\bibinfo {volume} {38}},\
  \bibinfo {pages} {L549} (\bibinfo {year} {2005})}\BibitemShut {NoStop}%
\bibitem [{\citenamefont {Borodin}\ \emph {et~al.}(2007)\citenamefont
  {Borodin}, \citenamefont {Ferrari}, \citenamefont {Pr\"ahofer},\ and\
  \citenamefont {Sasamoto}}]{Borodin.etal-JSP2007}%
  \BibitemOpen
  \bibfield  {author} {\bibinfo {author} {\bibfnamefont {A.}~\bibnamefont
  {Borodin}}, \bibinfo {author} {\bibfnamefont {P.~L.}\ \bibnamefont
  {Ferrari}}, \bibinfo {author} {\bibfnamefont {M.}~\bibnamefont
  {Pr\"ahofer}},\ and\ \bibinfo {author} {\bibfnamefont {T.}~\bibnamefont
  {Sasamoto}},\ }\bibfield  {title} {\bibinfo {title} {{Fluctuation properties
  of the TASEP with periodic initial configuration}},\ }\href
  {http://dx.doi.org/10.1007/s10955-007-9383-0} {\bibfield  {journal} {\bibinfo
   {journal} {J. Stat. Phys.}\ }\textbf {\bibinfo {volume} {129}},\ \bibinfo
  {pages} {1055} (\bibinfo {year} {2007})}\BibitemShut {NoStop}%
\bibitem [{\citenamefont {Widom}(2004)}]{Widom-JSP2004}%
  \BibitemOpen
  \bibfield  {author} {\bibinfo {author} {\bibfnamefont {H.}~\bibnamefont
  {Widom}},\ }\bibfield  {title} {\bibinfo {title} {On asymptotics for the
  {Airy} process},\ }\href {https://doi.org/10.1023/b:joss.0000022384.58696.61}
  {\bibfield  {journal} {\bibinfo  {journal} {J. Stat. Phys.}\ }\textbf
  {\bibinfo {volume} {115}},\ \bibinfo {pages} {1129} (\bibinfo {year}
  {2004})}\BibitemShut {NoStop}%
\bibitem [{\citenamefont {Basu}\ \emph {et~al.}(2023)\citenamefont {Basu},
  \citenamefont {Busani},\ and\ \citenamefont {Ferrari}}]{Basu.etal-CMP2023}%
  \BibitemOpen
  \bibfield  {author} {\bibinfo {author} {\bibfnamefont {R.}~\bibnamefont
  {Basu}}, \bibinfo {author} {\bibfnamefont {O.}~\bibnamefont {Busani}},\ and\
  \bibinfo {author} {\bibfnamefont {P.~L.}\ \bibnamefont {Ferrari}},\
  }\bibfield  {title} {\bibinfo {title} {On the exponent governing the
  correlation decay of the {Airy}$_1$ process},\ }\href
  {https://doi.org/10.1007/s00220-022-04544-1} {\bibfield  {journal} {\bibinfo
  {journal} {Commun. Math. Phys.}\ }\textbf {\bibinfo {volume} {398}},\
  \bibinfo {pages} {1171} (\bibinfo {year} {2023})}\BibitemShut {NoStop}%
\bibitem [{\citenamefont {Baik}\ \emph {et~al.}(2010)\citenamefont {Baik},
  \citenamefont {Ferrari},\ and\ \citenamefont
  {P\'ech\'e}}]{Baik.etal-CPAM2010}%
  \BibitemOpen
  \bibfield  {author} {\bibinfo {author} {\bibfnamefont {J.}~\bibnamefont
  {Baik}}, \bibinfo {author} {\bibfnamefont {P.~L.}\ \bibnamefont {Ferrari}},\
  and\ \bibinfo {author} {\bibfnamefont {S.}~\bibnamefont {P\'ech\'e}},\
  }\bibfield  {title} {\bibinfo {title} {Limit process of stationary {TASEP}
  near the characteristic line},\ }\href {http://dx.doi.org/10.1002/cpa.20316}
  {\bibfield  {journal} {\bibinfo  {journal} {Commun. Pure Appl. Math.}\
  }\textbf {\bibinfo {volume} {63}},\ \bibinfo {pages} {1017} (\bibinfo {year}
  {2010})}\BibitemShut {NoStop}%
\bibitem [{\citenamefont {De~Nardis}\ \emph {et~al.}(2023)\citenamefont
  {De~Nardis}, \citenamefont {Gopalakrishnan},\ and\ \citenamefont
  {Vasseur}}]{DeNardis.etal-PRL2023}%
  \BibitemOpen
  \bibfield  {author} {\bibinfo {author} {\bibfnamefont {J.}~\bibnamefont
  {De~Nardis}}, \bibinfo {author} {\bibfnamefont {S.}~\bibnamefont
  {Gopalakrishnan}},\ and\ \bibinfo {author} {\bibfnamefont {R.}~\bibnamefont
  {Vasseur}},\ }\bibfield  {title} {\bibinfo {title} {Nonlinear fluctuating
  hydrodynamics for {Kardar-Parisi-Zhang} scaling in isotropic spin chains},\
  }\href {https://doi.org/10.1103/PhysRevLett.131.197102} {\bibfield  {journal}
  {\bibinfo  {journal} {Phys. Rev. Lett.}\ }\textbf {\bibinfo {volume} {131}},\
  \bibinfo {pages} {197102} (\bibinfo {year} {2023})}\BibitemShut {NoStop}%
\bibitem [{Note1()}]{Note1}%
  \BibitemOpen
  \bibinfo {note} {Since the exact expression of the energy current for the
  KPLL model is unknown, here we borrow the expression for the Heisenberg model
  \cite {Zotos.etal-PRB1997}, which is expected to have a significant overlap
  with the true energy current of the KPLL model.}\BibitemShut {Stop}%
\bibitem [{\citenamefont {Pr\"ahofer}\ and\ \citenamefont
  {Spohn}(2000)}]{Prahofer.Spohn-PRL2000}%
  \BibitemOpen
  \bibfield  {author} {\bibinfo {author} {\bibfnamefont {M.}~\bibnamefont
  {Pr\"ahofer}}\ and\ \bibinfo {author} {\bibfnamefont {H.}~\bibnamefont
  {Spohn}},\ }\bibfield  {title} {\bibinfo {title} {Universal distributions for
  growth processes in $1+1$ dimensions and random matrices},\ }\href
  {https://doi.org/10.1103/PhysRevLett.84.4882} {\bibfield  {journal} {\bibinfo
   {journal} {Phys. Rev. Lett.}\ }\textbf {\bibinfo {volume} {84}},\ \bibinfo
  {pages} {4882} (\bibinfo {year} {2000})}\BibitemShut {NoStop}%
\bibitem [{\citenamefont {Ferrari}(2008)}]{Ferrari-JSM2008}%
  \BibitemOpen
  \bibfield  {author} {\bibinfo {author} {\bibfnamefont {P.~L.}\ \bibnamefont
  {Ferrari}},\ }\bibfield  {title} {\bibinfo {title} {Slow decorrelations in
  {Kardar-Parisi-Zhang} growth},\ }\href
  {http://stacks.iop.org/1742-5468/2008/i=07/a=P07022} {\bibfield  {journal}
  {\bibinfo  {journal} {J. Stat. Mech.}\ }\textbf {\bibinfo {volume} {2008}},\
  \bibinfo {pages} {P07022} (\bibinfo {year} {2008})}\BibitemShut {NoStop}%
\bibitem [{\citenamefont {Corwin}\ \emph {et~al.}(2012)\citenamefont {Corwin},
  \citenamefont {Ferrari},\ and\ \citenamefont
  {P\'ech\'e}}]{Corwin.etal-AIHPBPS2012}%
  \BibitemOpen
  \bibfield  {author} {\bibinfo {author} {\bibfnamefont {I.}~\bibnamefont
  {Corwin}}, \bibinfo {author} {\bibfnamefont {P.~L.}\ \bibnamefont
  {Ferrari}},\ and\ \bibinfo {author} {\bibfnamefont {S.}~\bibnamefont
  {P\'ech\'e}},\ }\bibfield  {title} {\bibinfo {title} {Universality of slow
  decorrelation in {KPZ} growth},\ }\href {https://doi.org/10.1214/11-AIHP440}
  {\bibfield  {journal} {\bibinfo  {journal} {Ann. Inst. H. Poincar\'e B
  Probab. Statist.}\ }\textbf {\bibinfo {volume} {48}},\ \bibinfo {pages} {134}
  (\bibinfo {year} {2012})}\BibitemShut {NoStop}%
\bibitem [{\citenamefont {Mendl}\ and\ \citenamefont
  {Spohn}(2016)}]{Mendl.Spohn-PRE2016}%
  \BibitemOpen
  \bibfield  {author} {\bibinfo {author} {\bibfnamefont {C.~B.}\ \bibnamefont
  {Mendl}}\ and\ \bibinfo {author} {\bibfnamefont {H.}~\bibnamefont {Spohn}},\
  }\bibfield  {title} {\bibinfo {title} {Searching for the {Tracy-Widom}
  distribution in nonequilibrium processes},\ }\href
  {https://doi.org/10.1103/PhysRevE.93.060101} {\bibfield  {journal} {\bibinfo
  {journal} {Phys. Rev. E}\ }\textbf {\bibinfo {volume} {93}},\ \bibinfo
  {pages} {060101(R)} (\bibinfo {year} {2016})}\BibitemShut {NoStop}%
\bibitem [{\citenamefont {Takeuchi}(2013)}]{Takeuchi-PRL2013}%
  \BibitemOpen
  \bibfield  {author} {\bibinfo {author} {\bibfnamefont {K.~A.}\ \bibnamefont
  {Takeuchi}},\ }\bibfield  {title} {\bibinfo {title} {Crossover from growing
  to stationary interfaces in the {Kardar-Parisi-Zhang} class},\ }\href
  {https://doi.org/10.1103/PhysRevLett.110.210604} {\bibfield  {journal}
  {\bibinfo  {journal} {Phys. Rev. Lett.}\ }\textbf {\bibinfo {volume} {110}},\
  \bibinfo {pages} {210604} (\bibinfo {year} {2013})}\BibitemShut {NoStop}%
\bibitem [{\citenamefont {Gopalakrishnan}\ \emph {et~al.}(2019)\citenamefont
  {Gopalakrishnan}, \citenamefont {Vasseur},\ and\ \citenamefont
  {Ware}}]{Gopalakrishnan.etal-PNAS2019}%
  \BibitemOpen
  \bibfield  {author} {\bibinfo {author} {\bibfnamefont {S.}~\bibnamefont
  {Gopalakrishnan}}, \bibinfo {author} {\bibfnamefont {R.}~\bibnamefont
  {Vasseur}},\ and\ \bibinfo {author} {\bibfnamefont {B.}~\bibnamefont
  {Ware}},\ }\bibfield  {title} {\bibinfo {title} {Anomalous relaxation and the
  high-temperature structure factor of {XXZ} spin chains},\ }\href
  {https://doi.org/10.1073/pnas.1906914116} {\bibfield  {journal} {\bibinfo
  {journal} {Proc. Natl. Acad. Sci. USA}\ }\textbf {\bibinfo {volume} {116}},\
  \bibinfo {pages} {16250} (\bibinfo {year} {2019})}\BibitemShut {NoStop}%
\bibitem [{\citenamefont {Cecile}\ \emph {et~al.}(2024)\citenamefont {Cecile},
  \citenamefont {De~Nardis},\ and\ \citenamefont
  {Ilievski}}]{PhysRevLett.132.130401}%
  \BibitemOpen
  \bibfield  {author} {\bibinfo {author} {\bibfnamefont {G.}~\bibnamefont
  {Cecile}}, \bibinfo {author} {\bibfnamefont {J.}~\bibnamefont {De~Nardis}},\
  and\ \bibinfo {author} {\bibfnamefont {E.}~\bibnamefont {Ilievski}},\
  }\bibfield  {title} {\bibinfo {title} {Squeezed ensembles and anomalous
  dynamic roughening in interacting integrable chains},\ }\href
  {https://doi.org/10.1103/PhysRevLett.132.130401} {\bibfield  {journal}
  {\bibinfo  {journal} {Phys. Rev. Lett.}\ }\textbf {\bibinfo {volume} {132}},\
  \bibinfo {pages} {130401} (\bibinfo {year} {2024})}\BibitemShut {NoStop}%
\bibitem [{Note2()}]{Note2}%
  \BibitemOpen
  \bibinfo {note} {After this initialization, we shift all spins slightly so
  that $\DOTSB \sum@ \slimits@ _j S_j^z(0)=0$ holds strictly.}\BibitemShut
  {Stop}%
\bibitem [{\citenamefont {Quastel}\ and\ \citenamefont
  {Remenik}(2019)}]{Quastel.Remenik-TAMS2019}%
  \BibitemOpen
  \bibfield  {author} {\bibinfo {author} {\bibfnamefont {J.}~\bibnamefont
  {Quastel}}\ and\ \bibinfo {author} {\bibfnamefont {D.}~\bibnamefont
  {Remenik}},\ }\bibfield  {title} {\bibinfo {title} {How flat is flat in
  random interface growth?},\ }\href {https://doi.org/10.1090/tran/7338}
  {\bibfield  {journal} {\bibinfo  {journal} {Trans. Am. Math. Soc.}\ }\textbf
  {\bibinfo {volume} {371}},\ \bibinfo {pages} {6047} (\bibinfo {year}
  {2019})}\BibitemShut {NoStop}%
\bibitem [{\citenamefont {{Levitov}}\ and\ \citenamefont
  {{Lesovik}}(1993)}]{Levitov1993}%
  \BibitemOpen
  \bibfield  {author} {\bibinfo {author} {\bibfnamefont {L.~S.}\ \bibnamefont
  {{Levitov}}}\ and\ \bibinfo {author} {\bibfnamefont {G.~B.}\ \bibnamefont
  {{Lesovik}}},\ }\bibfield  {title} {\bibinfo {title} {{Charge distribution in
  quantum shot noise}},\ }\href@noop {} {\bibfield  {journal} {\bibinfo
  {journal} {JETP Lett.}\ }\textbf {\bibinfo {volume} {58}},\ \bibinfo {pages}
  {230} (\bibinfo {year} {1993})}\BibitemShut {NoStop}%
\bibitem [{\citenamefont {Levitov}\ \emph {et~al.}(1996)\citenamefont
  {Levitov}, \citenamefont {Lee},\ and\ \citenamefont {Lesovik}}]{Levitov1996}%
  \BibitemOpen
  \bibfield  {author} {\bibinfo {author} {\bibfnamefont {L.~S.}\ \bibnamefont
  {Levitov}}, \bibinfo {author} {\bibfnamefont {H.}~\bibnamefont {Lee}},\ and\
  \bibinfo {author} {\bibfnamefont {G.~B.}\ \bibnamefont {Lesovik}},\
  }\bibfield  {title} {\bibinfo {title} {{Electron counting statistics and
  coherent states of electric current}},\ }\href
  {https://doi.org/10.1063/1.531672} {\bibfield  {journal} {\bibinfo  {journal}
  {J. Math. Phys.}\ }\textbf {\bibinfo {volume} {37}},\ \bibinfo {pages} {4845}
  (\bibinfo {year} {1996})}\BibitemShut {NoStop}%
\bibitem [{\citenamefont {Vidal}(2004)}]{Vidal-PRL2004}%
  \BibitemOpen
  \bibfield  {author} {\bibinfo {author} {\bibfnamefont {G.}~\bibnamefont
  {Vidal}},\ }\bibfield  {title} {\bibinfo {title} {Efficient simulation of
  one-dimensional quantum many-body systems},\ }\href
  {https://doi.org/10.1103/PhysRevLett.93.040502} {\bibfield  {journal}
  {\bibinfo  {journal} {Phys. Rev. Lett.}\ }\textbf {\bibinfo {volume} {93}},\
  \bibinfo {pages} {040502} (\bibinfo {year} {2004})}\BibitemShut {NoStop}%
\bibitem [{\citenamefont {Fishman}\ \emph {et~al.}(2022)\citenamefont
  {Fishman}, \citenamefont {White},\ and\ \citenamefont
  {Stoudenmire}}]{ITensor}%
  \BibitemOpen
  \bibfield  {author} {\bibinfo {author} {\bibfnamefont {M.}~\bibnamefont
  {Fishman}}, \bibinfo {author} {\bibfnamefont {S.~R.}\ \bibnamefont {White}},\
  and\ \bibinfo {author} {\bibfnamefont {E.~M.}\ \bibnamefont {Stoudenmire}},\
  }\bibfield  {title} {\bibinfo {title} {{The ITensor Software Library for
  Tensor Network Calculations}},\ }\href
  {https://doi.org/10.21468/SciPostPhysCodeb.4} {\bibfield  {journal} {\bibinfo
   {journal} {SciPost Phys. Codebases}\ ,\ \bibinfo {pages} {4}} (\bibinfo
  {year} {2022})}\BibitemShut {NoStop}%
\bibitem [{\citenamefont {Dauvergne}\ \emph {et~al.}(2022)\citenamefont
  {Dauvergne}, \citenamefont {Ortmann},\ and\ \citenamefont {Vir\'ag}}]{DOV22}%
  \BibitemOpen
  \bibfield  {author} {\bibinfo {author} {\bibfnamefont {D.}~\bibnamefont
  {Dauvergne}}, \bibinfo {author} {\bibfnamefont {J.}~\bibnamefont {Ortmann}},\
  and\ \bibinfo {author} {\bibfnamefont {B.}~\bibnamefont {Vir\'ag}},\
  }\bibfield  {title} {\bibinfo {title} {The directed landscape},\ }\href@noop
  {} {\bibfield  {journal} {\bibinfo  {journal} {Acta Mathematica}\ }\textbf
  {\bibinfo {volume} {229}},\ \bibinfo {pages} {201} (\bibinfo {year}
  {2022})}\BibitemShut {NoStop}%
\bibitem [{\citenamefont {Emrah}\ \emph {et~al.}(2020)\citenamefont {Emrah},
  \citenamefont {Janjigian},\ and\ \citenamefont
  {Sepp{\"a}l{\"a}inen}}]{EJS20}%
  \BibitemOpen
  \bibfield  {author} {\bibinfo {author} {\bibfnamefont {E.}~\bibnamefont
  {Emrah}}, \bibinfo {author} {\bibfnamefont {C.}~\bibnamefont {Janjigian}},\
  and\ \bibinfo {author} {\bibfnamefont {T.}~\bibnamefont
  {Sepp{\"a}l{\"a}inen}},\ }\bibfield  {title} {\bibinfo {title} {Right-tail
  moderate deviations in the exponential last-passage percolation},\
  }\href@noop {} {\bibfield  {journal} {\bibinfo  {journal} {arXiv:2004.04285}\
  } (\bibinfo {year} {2020})}\BibitemShut {NoStop}%
\bibitem [{\citenamefont {Cator}\ and\ \citenamefont {Pimentel}(2015)}]{CP15b}%
  \BibitemOpen
  \bibfield  {author} {\bibinfo {author} {\bibfnamefont {E.}~\bibnamefont
  {Cator}}\ and\ \bibinfo {author} {\bibfnamefont {L.}~\bibnamefont
  {Pimentel}},\ }\bibfield  {title} {\bibinfo {title} {On the local
  fluctuations of last-passage percolation models},\ }\href@noop {} {\bibfield
  {journal} {\bibinfo  {journal} {Stoch. Proc. Appl.}\ }\textbf {\bibinfo
  {volume} {125}},\ \bibinfo {pages} {538} (\bibinfo {year}
  {2015})}\BibitemShut {NoStop}%
\bibitem [{\citenamefont {Ferrari}\ and\ \citenamefont {Occelli}(2018)}]{FO17}%
  \BibitemOpen
  \bibfield  {author} {\bibinfo {author} {\bibfnamefont {P.}~\bibnamefont
  {Ferrari}}\ and\ \bibinfo {author} {\bibfnamefont {A.}~\bibnamefont
  {Occelli}},\ }\bibfield  {title} {\bibinfo {title} {{Universality of the GOE
  Tracy-Widom distribution for TASEP with arbitrary particle density}},\
  }\href@noop {} {\bibfield  {journal} {\bibinfo  {journal} {Eletron. J.
  Probab.}\ }\textbf {\bibinfo {volume} {23}},\ \bibinfo {pages} {1} (\bibinfo
  {year} {2018})}\BibitemShut {NoStop}%
\bibitem [{\citenamefont {Tracy}\ and\ \citenamefont {Widom}(1994)}]{TW94}%
  \BibitemOpen
  \bibfield  {author} {\bibinfo {author} {\bibfnamefont {C.}~\bibnamefont
  {Tracy}}\ and\ \bibinfo {author} {\bibfnamefont {H.}~\bibnamefont {Widom}},\
  }\bibfield  {title} {\bibinfo {title} {{Level-spacing distributions and the
  Airy kernel}},\ }\href@noop {} {\bibfield  {journal} {\bibinfo  {journal}
  {Commun. Math. Phys.}\ }\textbf {\bibinfo {volume} {159}},\ \bibinfo {pages}
  {151} (\bibinfo {year} {1994})}\BibitemShut {NoStop}%
\bibitem [{\citenamefont {Quastel}\ and\ \citenamefont {Remenik}(2013)}]{QR12}%
  \BibitemOpen
  \bibfield  {author} {\bibinfo {author} {\bibfnamefont {J.}~\bibnamefont
  {Quastel}}\ and\ \bibinfo {author} {\bibfnamefont {D.}~\bibnamefont
  {Remenik}},\ }\bibfield  {title} {\bibinfo {title} {{Local behavior and
  hitting probabilities of the Airy$_1$ process}},\ }\href@noop {} {\bibfield
  {journal} {\bibinfo  {journal} {Prob. Theory Relat. Fields}\ }\textbf
  {\bibinfo {volume} {157}},\ \bibinfo {pages} {605} (\bibinfo {year}
  {2013})}\BibitemShut {NoStop}%
\bibitem [{\citenamefont {Ferrari}\ and\ \citenamefont {Spohn}(2006)}]{FS05a}%
  \BibitemOpen
  \bibfield  {author} {\bibinfo {author} {\bibfnamefont {P.}~\bibnamefont
  {Ferrari}}\ and\ \bibinfo {author} {\bibfnamefont {H.}~\bibnamefont
  {Spohn}},\ }\bibfield  {title} {\bibinfo {title} {Scaling limit for the
  space-time covariance of the stationary totally asymmetric simple exclusion
  process},\ }\href@noop {} {\bibfield  {journal} {\bibinfo  {journal} {Commun.
  Math. Phys.}\ }\textbf {\bibinfo {volume} {265}},\ \bibinfo {pages} {1}
  (\bibinfo {year} {2006})}\BibitemShut {NoStop}%
\bibitem [{\citenamefont {Zotos}\ \emph {et~al.}(1997)\citenamefont {Zotos},
  \citenamefont {Naef},\ and\ \citenamefont {Prelovsek}}]{Zotos.etal-PRB1997}%
  \BibitemOpen
  \bibfield  {author} {\bibinfo {author} {\bibfnamefont {X.}~\bibnamefont
  {Zotos}}, \bibinfo {author} {\bibfnamefont {F.}~\bibnamefont {Naef}},\ and\
  \bibinfo {author} {\bibfnamefont {P.}~\bibnamefont {Prelovsek}},\ }\bibfield
  {title} {\bibinfo {title} {Transport and conservation laws},\ }\href
  {https://doi.org/10.1103/PhysRevB.55.11029} {\bibfield  {journal} {\bibinfo
  {journal} {Phys. Rev. B}\ }\textbf {\bibinfo {volume} {55}},\ \bibinfo
  {pages} {11029} (\bibinfo {year} {1997})}\BibitemShut {NoStop}%
\end{thebibliography}%


\begin{thebibliography}{16}%
\makeatletter
\providecommand \@ifxundefined [1]{%
 \@ifx{#1\undefined}
}%
\providecommand \@ifnum [1]{%
 \ifnum #1\expandafter \@firstoftwo
 \else \expandafter \@secondoftwo
 \fi
}%
\providecommand \@ifx [1]{%
 \ifx #1\expandafter \@firstoftwo
 \else \expandafter \@secondoftwo
 \fi
}%
\providecommand \natexlab [1]{#1}%
\providecommand \enquote  [1]{``#1''}%
\providecommand \bibnamefont  [1]{#1}%
\providecommand \bibfnamefont [1]{#1}%
\providecommand \citenamefont [1]{#1}%
\providecommand \href@noop [0]{\@secondoftwo}%
\providecommand \href [0]{\begingroup \@sanitize@url \@href}%
\providecommand \@href[1]{\@@startlink{#1}\@@href}%
\providecommand \@@href[1]{\endgroup#1\@@endlink}%
\providecommand \@sanitize@url [0]{\catcode `\\12\catcode `\$12\catcode
  `\&12\catcode `\#12\catcode `\^12\catcode `\_12\catcode `\%12\relax}%
\providecommand \@@startlink[1]{}%
\providecommand \@@endlink[0]{}%
\providecommand \url  [0]{\begingroup\@sanitize@url \@url }%
\providecommand \@url [1]{\endgroup\@href {#1}{\urlprefix }}%
\providecommand \urlprefix  [0]{URL }%
\providecommand \Eprint [0]{\href }%
\providecommand \doibase [0]{https://doi.org/}%
\providecommand \selectlanguage [0]{\@gobble}%
\providecommand \bibinfo  [0]{\@secondoftwo}%
\providecommand \bibfield  [0]{\@secondoftwo}%
\providecommand \translation [1]{[#1]}%
\providecommand \BibitemOpen [0]{}%
\providecommand \bibitemStop [0]{}%
\providecommand \bibitemNoStop [0]{.\EOS\space}%
\providecommand \EOS [0]{\spacefactor3000\relax}%
\providecommand \BibitemShut  [1]{\csname bibitem#1\endcsname}%
\let\auto@bib@innerbib\@empty
\bibitem [{\citenamefont {Krajnik}\ and\ \citenamefont
  {Prosen}(2020)}]{Krajnik.Prosen-JSP2020}%
  \BibitemOpen
  \bibfield  {author} {\bibinfo {author} {\bibfnamefont {{\v{Z}}.}~\bibnamefont
  {Krajnik}}\ and\ \bibinfo {author} {\bibfnamefont {T.}~\bibnamefont
  {Prosen}},\ }\bibfield  {title} {\bibinfo {title} {{Kardar-Parisi-Zhang}
  physics in integrable rotationally symmetric dynamics on discrete space-time
  lattice},\ }\href {https://doi.org/10.1007/s10955-020-02523-1} {\bibfield
  {journal} {\bibinfo  {journal} {J. Stat. Phys.}\ }\textbf {\bibinfo {volume}
  {179}},\ \bibinfo {pages} {110} (\bibinfo {year} {2020})}\BibitemShut
  {NoStop}%
\bibitem [{\citenamefont {Krajnik}\ \emph {et~al.}(2021)\citenamefont
  {Krajnik}, \citenamefont {Ilievski}, \citenamefont {Prosen},\ and\
  \citenamefont {Pasquier}}]{Krajnik.etal-SPP2021}%
  \BibitemOpen
  \bibfield  {author} {\bibinfo {author} {\bibfnamefont {{\v{Z}}.}~\bibnamefont
  {Krajnik}}, \bibinfo {author} {\bibfnamefont {E.}~\bibnamefont {Ilievski}},
  \bibinfo {author} {\bibfnamefont {T.}~\bibnamefont {Prosen}},\ and\ \bibinfo
  {author} {\bibfnamefont {V.}~\bibnamefont {Pasquier}},\ }\bibfield  {title}
  {\bibinfo {title} {Anisotropic {Landau-Lifshitz} model in discrete
  space-time},\ }\href {https://doi.org/10.21468/scipostphys.11.3.051}
  {\bibfield  {journal} {\bibinfo  {journal} {SciPost Phys.}\ }\textbf
  {\bibinfo {volume} {11}},\ \bibinfo {pages} {051} (\bibinfo {year}
  {2021})}\BibitemShut {NoStop}%
\bibitem [{\citenamefont {Krajnik}\ \emph {et~al.}(2022)\citenamefont
  {Krajnik}, \citenamefont {Ilievski},\ and\ \citenamefont
  {Prosen}}]{Krajnik.etal-PRL2022}%
  \BibitemOpen
  \bibfield  {author} {\bibinfo {author} {\bibfnamefont {{\v{Z}}.}~\bibnamefont
  {Krajnik}}, \bibinfo {author} {\bibfnamefont {E.}~\bibnamefont {Ilievski}},\
  and\ \bibinfo {author} {\bibfnamefont {T.}~\bibnamefont {Prosen}},\
  }\bibfield  {title} {\bibinfo {title} {Absence of normal fluctuations in an
  integrable magnet},\ }\href {https://doi.org/10.1103/PhysRevLett.128.090604}
  {\bibfield  {journal} {\bibinfo  {journal} {Phys. Rev. Lett.}\ }\textbf
  {\bibinfo {volume} {128}},\ \bibinfo {pages} {090604} (\bibinfo {year}
  {2022})}\BibitemShut {NoStop}%
\bibitem [{\citenamefont {Ishimori}(1982)}]{Ishimori-JPSJ1982}%
  \BibitemOpen
  \bibfield  {author} {\bibinfo {author} {\bibfnamefont {Y.}~\bibnamefont
  {Ishimori}},\ }\bibfield  {title} {\bibinfo {title} {An integrable classical
  spin chain},\ }\href {https://doi.org/10.1143/jpsj.51.3417} {\bibfield
  {journal} {\bibinfo  {journal} {J. Phys. Soc. Jpn.}\ }\textbf {\bibinfo
  {volume} {51}},\ \bibinfo {pages} {3417} (\bibinfo {year}
  {1982})}\BibitemShut {NoStop}%
\bibitem [{\citenamefont {Baik}\ \emph {et~al.}(2010)\citenamefont {Baik},
  \citenamefont {Ferrari},\ and\ \citenamefont
  {P\'ech\'e}}]{Baik.etal-CPAM2010}%
  \BibitemOpen
  \bibfield  {author} {\bibinfo {author} {\bibfnamefont {J.}~\bibnamefont
  {Baik}}, \bibinfo {author} {\bibfnamefont {P.~L.}\ \bibnamefont {Ferrari}},\
  and\ \bibinfo {author} {\bibfnamefont {S.}~\bibnamefont {P\'ech\'e}},\
  }\bibfield  {title} {\bibinfo {title} {Limit process of stationary {TASEP}
  near the characteristic line},\ }\href {http://dx.doi.org/10.1002/cpa.20316}
  {\bibfield  {journal} {\bibinfo  {journal} {Commun. Pure Appl. Math.}\
  }\textbf {\bibinfo {volume} {63}},\ \bibinfo {pages} {1017} (\bibinfo {year}
  {2010})}\BibitemShut {NoStop}%
\bibitem [{\citenamefont {Dauvergne}\ \emph {et~al.}(2022)\citenamefont
  {Dauvergne}, \citenamefont {Ortmann},\ and\ \citenamefont {Vir\'ag}}]{DOV22}%
  \BibitemOpen
  \bibfield  {author} {\bibinfo {author} {\bibfnamefont {D.}~\bibnamefont
  {Dauvergne}}, \bibinfo {author} {\bibfnamefont {J.}~\bibnamefont {Ortmann}},\
  and\ \bibinfo {author} {\bibfnamefont {B.}~\bibnamefont {Vir\'ag}},\
  }\bibfield  {title} {\bibinfo {title} {The directed landscape},\ }\href@noop
  {} {\bibfield  {journal} {\bibinfo  {journal} {Acta Mathematica}\ }\textbf
  {\bibinfo {volume} {229}},\ \bibinfo {pages} {201} (\bibinfo {year}
  {2022})}\BibitemShut {NoStop}%
\bibitem [{\citenamefont {Emrah}\ \emph {et~al.}(2020)\citenamefont {Emrah},
  \citenamefont {Janjigian},\ and\ \citenamefont
  {Sepp{\"a}l{\"a}inen}}]{EJS20}%
  \BibitemOpen
  \bibfield  {author} {\bibinfo {author} {\bibfnamefont {E.}~\bibnamefont
  {Emrah}}, \bibinfo {author} {\bibfnamefont {C.}~\bibnamefont {Janjigian}},\
  and\ \bibinfo {author} {\bibfnamefont {T.}~\bibnamefont
  {Sepp{\"a}l{\"a}inen}},\ }\bibfield  {title} {\bibinfo {title} {Right-tail
  moderate deviations in the exponential last-passage percolation},\
  }\href@noop {} {\bibfield  {journal} {\bibinfo  {journal} {arXiv:2004.04285}\
  } (\bibinfo {year} {2020})}\BibitemShut {NoStop}%
\bibitem [{\citenamefont {Cator}\ and\ \citenamefont {Pimentel}(2015)}]{CP15b}%
  \BibitemOpen
  \bibfield  {author} {\bibinfo {author} {\bibfnamefont {E.}~\bibnamefont
  {Cator}}\ and\ \bibinfo {author} {\bibfnamefont {L.}~\bibnamefont
  {Pimentel}},\ }\bibfield  {title} {\bibinfo {title} {On the local
  fluctuations of last-passage percolation models},\ }\href@noop {} {\bibfield
  {journal} {\bibinfo  {journal} {Stoch. Proc. Appl.}\ }\textbf {\bibinfo
  {volume} {125}},\ \bibinfo {pages} {538} (\bibinfo {year}
  {2015})}\BibitemShut {NoStop}%
\bibitem [{\citenamefont {Pr\"ahofer}\ and\ \citenamefont
  {Spohn}(2002)}]{Prahofer.Spohn-JSP2002}%
  \BibitemOpen
  \bibfield  {author} {\bibinfo {author} {\bibfnamefont {M.}~\bibnamefont
  {Pr\"ahofer}}\ and\ \bibinfo {author} {\bibfnamefont {H.}~\bibnamefont
  {Spohn}},\ }\bibfield  {title} {\bibinfo {title} {{Scale Invariance of the
  PNG Droplet and the Airy Process}},\ }\href
  {http://dx.doi.org/10.1023/A:1019791415147} {\bibfield  {journal} {\bibinfo
  {journal} {J. Stat. Phys.}\ }\textbf {\bibinfo {volume} {108}},\ \bibinfo
  {pages} {1071} (\bibinfo {year} {2002})}\BibitemShut {NoStop}%
\bibitem [{\citenamefont {Ferrari}\ and\ \citenamefont {Occelli}(2018)}]{FO17}%
  \BibitemOpen
  \bibfield  {author} {\bibinfo {author} {\bibfnamefont {P.}~\bibnamefont
  {Ferrari}}\ and\ \bibinfo {author} {\bibfnamefont {A.}~\bibnamefont
  {Occelli}},\ }\bibfield  {title} {\bibinfo {title} {{Universality of the GOE
  Tracy-Widom distribution for TASEP with arbitrary particle density}},\
  }\href@noop {} {\bibfield  {journal} {\bibinfo  {journal} {Eletron. J.
  Probab.}\ }\textbf {\bibinfo {volume} {23}},\ \bibinfo {pages} {1} (\bibinfo
  {year} {2018})}\BibitemShut {NoStop}%
\bibitem [{\citenamefont {Tracy}\ and\ \citenamefont {Widom}(1994)}]{TW94}%
  \BibitemOpen
  \bibfield  {author} {\bibinfo {author} {\bibfnamefont {C.}~\bibnamefont
  {Tracy}}\ and\ \bibinfo {author} {\bibfnamefont {H.}~\bibnamefont {Widom}},\
  }\bibfield  {title} {\bibinfo {title} {{Level-spacing distributions and the
  Airy kernel}},\ }\href@noop {} {\bibfield  {journal} {\bibinfo  {journal}
  {Commun. Math. Phys.}\ }\textbf {\bibinfo {volume} {159}},\ \bibinfo {pages}
  {151} (\bibinfo {year} {1994})}\BibitemShut {NoStop}%
\bibitem [{\citenamefont {Quastel}\ and\ \citenamefont {Remenik}(2013)}]{QR12}%
  \BibitemOpen
  \bibfield  {author} {\bibinfo {author} {\bibfnamefont {J.}~\bibnamefont
  {Quastel}}\ and\ \bibinfo {author} {\bibfnamefont {D.}~\bibnamefont
  {Remenik}},\ }\bibfield  {title} {\bibinfo {title} {{Local behavior and
  hitting probabilities of the Airy$_1$ process}},\ }\href@noop {} {\bibfield
  {journal} {\bibinfo  {journal} {Prob. Theory Relat. Fields}\ }\textbf
  {\bibinfo {volume} {157}},\ \bibinfo {pages} {605} (\bibinfo {year}
  {2013})}\BibitemShut {NoStop}%
\bibitem [{\citenamefont {Baik}\ and\ \citenamefont
  {Rains}(2000)}]{Baik.Rains-JSP2000}%
  \BibitemOpen
  \bibfield  {author} {\bibinfo {author} {\bibfnamefont {J.}~\bibnamefont
  {Baik}}\ and\ \bibinfo {author} {\bibfnamefont {E.~M.}\ \bibnamefont
  {Rains}},\ }\bibfield  {title} {\bibinfo {title} {Limiting distributions for
  a polynuclear growth model with external sources},\ }\href
  {http://dx.doi.org/10.1023/A:1018615306992} {\bibfield  {journal} {\bibinfo
  {journal} {J. Stat. Phys.}\ }\textbf {\bibinfo {volume} {100}},\ \bibinfo
  {pages} {523} (\bibinfo {year} {2000})}\BibitemShut {NoStop}%
\bibitem [{\citenamefont {Ferrari}\ and\ \citenamefont {Spohn}(2006)}]{FS05a}%
  \BibitemOpen
  \bibfield  {author} {\bibinfo {author} {\bibfnamefont {P.}~\bibnamefont
  {Ferrari}}\ and\ \bibinfo {author} {\bibfnamefont {H.}~\bibnamefont
  {Spohn}},\ }\bibfield  {title} {\bibinfo {title} {Scaling limit for the
  space-time covariance of the stationary totally asymmetric simple exclusion
  process},\ }\href@noop {} {\bibfield  {journal} {\bibinfo  {journal} {Commun.
  Math. Phys.}\ }\textbf {\bibinfo {volume} {265}},\ \bibinfo {pages} {1}
  (\bibinfo {year} {2006})}\BibitemShut {NoStop}%
\bibitem [{\citenamefont {Pr{\"a}hofer}\ and\ \citenamefont
  {Spohn}(2002)}]{PS01}%
  \BibitemOpen
  \bibfield  {author} {\bibinfo {author} {\bibfnamefont {M.}~\bibnamefont
  {Pr{\"a}hofer}}\ and\ \bibinfo {author} {\bibfnamefont {H.}~\bibnamefont
  {Spohn}},\ }\bibfield  {title} {\bibinfo {title} {Current fluctuations for
  the totally asymmetric simple exclusion process},\ }in\ \href@noop {} {\emph
  {\bibinfo {booktitle} {In and out of equilibrium}}},\ \bibinfo {series and
  number} {Progress in Probability},\ \bibinfo {editor} {edited by\ \bibinfo
  {editor} {\bibfnamefont {V.}~\bibnamefont {Sidoravicius}}}\ (\bibinfo
  {publisher} {Birkh{\"a}user},\ \bibinfo {year} {2002})\BibitemShut {NoStop}%
\bibitem [{\citenamefont {Takeuchi}(2013)}]{Takeuchi-PRL2013}%
  \BibitemOpen
  \bibfield  {author} {\bibinfo {author} {\bibfnamefont {K.~A.}\ \bibnamefont
  {Takeuchi}},\ }\bibfield  {title} {\bibinfo {title} {Crossover from growing
  to stationary interfaces in the {Kardar-Parisi-Zhang} class},\ }\href
  {https://doi.org/10.1103/PhysRevLett.110.210604} {\bibfield  {journal}
  {\bibinfo  {journal} {Phys. Rev. Lett.}\ }\textbf {\bibinfo {volume} {110}},\
  \bibinfo {pages} {210604} (\bibinfo {year} {2013})}\BibitemShut {NoStop}%
\end{thebibliography}%

\end{document}


\title{Supplemental Material for ``Partial yet definite emergence of the Kardar-Parisi-Zhang class in isotropic spin chains''}

\author{Kazumasa A. Takeuchi}
\email{kat@kaztake.org}
\affiliation{Department of Physics,\! The University of Tokyo,\! 7-3-1 Hongo,\! Bunkyo-ku,\! Tokyo 113-0033,\! Japan}%
\affiliation{Institute for Physics of Intelligence,\! The University of Tokyo,\! 7-3-1 Hongo,\! Bunkyo-ku,\! Tokyo 113-0033,\! Japan}%

\author{Kazuaki Takasan}
\affiliation{Department of Physics,\! The University of Tokyo,\! 7-3-1 Hongo,\! Bunkyo-ku,\! Tokyo 113-0033,\! Japan}%

\author{Ofer Busani}
\affiliation{School of Mathematics, University of Edinburgh, James Clerk Maxwell Building, Peter Guthrie Tait Road, Edinburgh, EH9 3FD, UK}%

\author{Patrik L. Ferrari}
\affiliation{Institute for Applied Mathematics, Bonn University, Endenicher Allee 60, 53115 Bonn, Germany}%


\author{Romain Vasseur}
\affiliation{Department of Theoretical Physics, University of Geneva, 24 quai Ernest-Ansermet, 1211 Gen\`eve, Switzerland}

\affiliation{Department of Physics, University of Massachusetts, Amherst, MA 01003, USA}%

\author{Jacopo De Nardis}
\affiliation{Laboratoire de Physique Th\'eorique et Mod\'elisation, CNRS UMR 8089,
CY Cergy Paris Universit\'e, 95302 Cergy-Pontoise Cedex, France}%

\date{\today}

\maketitle

\onecolumngrid
\section{Supplemental Text 1: Integrable discrete Landau-Lifshitz magnet}

\begin{figure}[b!]
\includegraphics[width=0.9\hsize,clip]{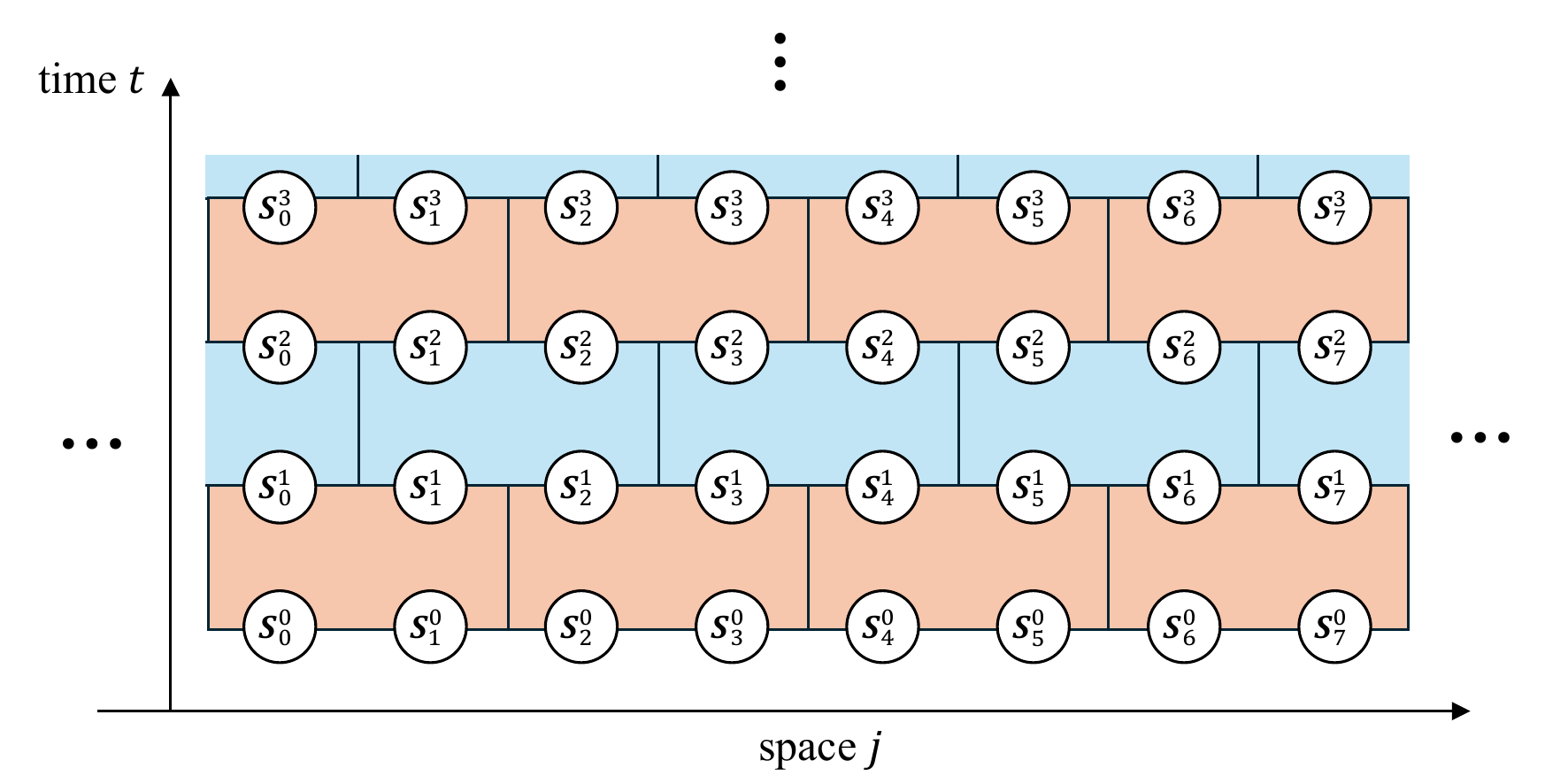}
\centering
\caption{
Schematic of the brick-layer space-time lattice of the KPLL model \cite{Krajnik.Prosen-JSP2020,Krajnik.etal-SPP2021,Krajnik.etal-PRL2022}.
}
\label{figS:LL}
\end{figure}

Here we describe the definition of the integrable discrete Landau-Lifshitz magnet introduced by Krajnik and Prosen \cite{Krajnik.Prosen-JSP2020,Krajnik.etal-SPP2021,Krajnik.etal-PRL2022} (the KPLL model), which we mainly use in this Letter.
The KPLL model is defined on a space-time lattice, as sketched in \figref{figS:LL}.
Time is discretized as $t = n\Delta t$ with $n \in \mathbb{Z}$.
Each site has a classic spin, $\bm{S}_j^n$ ($=\bm{S}_j(t)$ of the main text), which is a unit vector, and the periodic boundary is assumed.
As sketched in \figref{figS:LL} as a brick layer, the time evolution takes place pairwise, alternately, as follows:
\begin{equation} 
(\bm{S}_{2i}^{2m+1}, \bm{S}_{2i+1}^{2m+1}) = \Phi(\bm{S}_{2i}^{2m}, \bm{S}_{2i+1}^{2m}), \qquad
(\bm{S}_{2i-1}^{2m+2}, \bm{S}_{2i}^{2m+2}) = \Phi(\bm{S}_{2i-1}^{2m+1}, \bm{S}_{2i}^{2m+1}),
\end{equation}
with $i, m \in \mathbb{Z}$ and
\begin{equation}
    \Phi(\bm{S}_1, \bm{S}_2) \equiv \frac{1}{s^2 + \Delta t^2}(s^2 \bm{S}_1 + \Delta t^2 \bm{S}_2 + \Delta t \bm{S}_1 \times \bm{S}_2, \quad s^2 \bm{S}_2 + \Delta t^2 \bm{S}_1 + \Delta t \bm{S}_2 \times \bm{S}_1), \qquad s^2 \equiv \frac{1}{2}(1+ \bm{S}_1 \cdot \bm{S}_2).
\end{equation}
In the continuous time limit $\Delta t \to 0$, this model converges to the Ishimori chain \cite{Ishimori-JPSJ1982} [\eqref{M-eq:LL}]
\begin{equation}
\prt{\bm{S}_j}{t} = \{\bm{S}_j, H\} = \frac{\bm{S}_j \times \bm{S}_{j-1}}{1+\bm{S}_j \cdot \bm{S}_{j-1}} + \frac{\bm{S}_j \times \bm{S}_{j+1}}{1+\bm{S}_j \cdot \bm{S}_{j+1}},
\end{equation}
with Hamiltonian $H = \sum_j \log(1+\bm{S}_j \cdot \bm{S}_{j+1})$, a well-known integrable variant of the lattice Landau-Lifshitz model.


The KPLL model satisfies the following continuity equation for the $z$-component of the spin, $S_j^n \equiv (\bm{S}_j^n)^z$:
\begin{equation}
S_{2i}^{2m+2} - S_{2i}^{2m} = -J_{2i+\frac12}^{2m+\frac12} + J_{2i-\frac12}^{2m+\frac32}, \qquad
S_{2i+1}^{2m+2} - S_{2i+1}^{2m} = -J_{2i+\frac32}^{2m+\frac32} + J_{2i+\frac12}^{2m+\frac12},
\end{equation}
with the spin current (integrated over the time step $\Delta t$)
\begin{equation}
J_{2i+\frac12}^{2m+\frac12} = S_{2i+1}^{2m+1} - S_{2i+1}^{2m} = -(S_{2i}^{2m+1} - S_{2i}^{2m}), \qquad
J_{2i-\frac12}^{2m+\frac32} = S_{2i}^{2m+2} - S_{2i}^{2m+1} = -(S_{2i-1}^{2m+2} - S_{2i-1}^{2m+1}).
\end{equation}
Of course, by setting $J_{2i+\frac12}^{2m+\frac32} = J_{2i-\frac12}^{2m+\frac12} = 0$, we also have the following simplest expression of the continuity equation: 
\begin{equation}
    S_j^{n+1} - S_j^n = -J_{j+\frac12}^{n+\frac12} + J_{j-\frac12}^{n+\frac12}.
\end{equation}
With this spin current, the magnetization transfer (the integrated spin current) discussed in the main text is defined by
\begin{equation}
    h_{j+\frac12}^n \equiv \sum_{n'=0}^{n-1} J_{j+\frac12}^{n'+\frac12},
\end{equation}
which is denoted by $h(x,t) = \int_0^t J(x,t') dt'$ in the main text, with $x=j+\frac12$ and $t=n\Delta t$.
The expression for the case with a finite energy current, \eqref{M-eq:height}, can also be obtained straightforwardly, as follows:
\begin{equation}
    h_{j+\frac12}^n \equiv \sum_{n'=0}^{n-1} J_{j+\frac12}^{n'+\frac12} - \sum_{j' \in \mathcal{J}} S_{j'}^0
\end{equation}
where $\mathcal{J}$ is the set of integers between $j+\frac12$ and $j-vt+\frac12$. 


\section{Supplemental Text 2: Short time covariance of $\mathcal{A}_0$}
\twocolumngrid

The Airy$_{\rm stat}$ process was introduced in~\cite{Baik.etal-CPAM2010} and it is given by
\begin{equation}\label{eqS1}
\begin{aligned}
\mathcal{A}_{\rm stat}(u)&=\max_{v}\{\sqrt{2}B(v)+\LL(v,0;u,1)\},\\
\mathcal{A}_{\rm stat}(0)&=\max_{w}\{\sqrt{2}B(w)+\LL(w,0;0,1)\},
\end{aligned}
\end{equation}
where $\LL$ is the directed landscape~\cite{DOV22} and $B(u)$ is a standard two-sided Brownian motion (i.e., starting at ``time'' $u=0$ and with diffusivity constant $1$).
This process is not stationary, but its increments are stationary, with
\begin{equation}
\mathcal{A}_{\rm stat}(u)-\mathcal{A}_{\rm stat}(0)\stackrel{(d)}{=}\sqrt{2}B(u).
\end{equation}

The stationary version of it, that we denote by $\mathcal{A}_0(u)$ and name  \emph{Airy$_0$ process}, is given by
\begin{equation}
\mathcal{A}_0(u):=\max_{v}\{\LL(v,0;u,1)+\sqrt{2}B(v)-\sqrt{2}B(u)\}.
\end{equation}
Our goal is to prove the following result.
\begin{theorem}\label{thmCovA0} For small $u$ we have
\begin{equation}
    \Cov(\mathcal{A}_0(0);\mathcal{A}_0(u)) =\Var(\mathcal{A}_0(0))-2u+o(u).
\end{equation}
\end{theorem}
\begin{proof}
First of all, note that
\begin{equation}\label{eqS4}
\mathcal{A}_0(0)=\mathcal{A}_{\rm stat}(0),\quad \mathcal{A}_0(u)=\mathcal{A}_{\rm stat}(u)-\sqrt{2}B(u).
\end{equation}

The covariance of $\mathcal{A}_0$ between times $0$ and $u$ can be decomposed as
\begin{multline}\label{eqS6}
\Cov(\mathcal{A}_0(0);\mathcal{A}_0(u)) = \tfrac12 \Var(\mathcal{A}_0(0)) \\+ \tfrac12 \Var(\mathcal{A}_0(u))
-\tfrac12 \Var(\mathcal{A}_0(u)-\mathcal{A}_0(0)).
\end{multline}
The last term can be rewritten using \eqref{eqS4} as
\begin{multline}
\Var(\mathcal{A}_0(u)-\mathcal{A}_0(0)) = \Var(\mathcal{A}_{\rm stat}(u)-\mathcal{A}_{\rm stat}(0))\\
+\Var(\sqrt{2} B(u))
-2\Cov(\mathcal{A}_{\rm stat}(u)-\mathcal{A}_{\rm stat}(0);\sqrt{2}B(u))
\end{multline}
Since the distribution of $\mathcal{A}_0(u)$ is independent of $u$ and $\Var(\sqrt{2} B(u))=2u$, we obtain
\begin{multline}
\Cov(\mathcal{A}_0(0);\mathcal{A}_0(u)) =\Var(\mathcal{A}_0(0))-2u\\
+\Cov(\mathcal{A}_{\rm stat}(u)-\mathcal{A}_{\rm stat}(0);\sqrt{2}B(u)).
\end{multline}
In Lemma~\ref{thmSmallU} we show that the last term is $o(u)$, completing the proof.
\end{proof}

\begin{lemma}\label{thmSmallU} There exists a constant $C$ such that for all $0\leq u\leq 1$ we have
\begin{multline}
|\Cov(\mathcal{A}_{\rm stat}(u)-\mathcal{A}_{\rm stat}(0);\sqrt{2} B(u))|\\
\leq C u^{3/2} \left(1+(\ln(1/u))^{1/3}\right).
\end{multline}
\end{lemma}
Here, of course, the threshold $1$ can be replaced by any arbitrary $u_0>0$ by adapting the constant $C$ appropriately.
\begin{proof}[Proof of Lemma~\ref{thmSmallU}]
Let us denote by $v_0$ (resp.\ $w_0$) the argmax of the variational formulas in \eqref{eqS1}. Consider the event
\begin{equation}
G_M=\{-M\leq v_0,w_0\leq M+u\}.
\end{equation}
Using Theorem~2.5 and Proposition~2.7 of~\cite{EJS20} we get
\begin{equation}\label{eqGmC}
\Pb(G_M^c)\leq C e^{-c M^3}
\end{equation}
for some constants $C,c>0$.
Assume that $G_M$ occurs. Then the geodesic from $(-M,0)$ to $(u,1)$ crosses that of the stationary process $\mathcal{A}_{\rm stat}(0)$, and the geodesic from $(M,0)$ to $(0,1)$ crosses that of $\mathcal{A}_{\rm stat}(u)$. Using the comparison inequality approach first introduced in~\cite{CP15b} we get
\begin{multline}\label{eqComparison}
\LL(-M,0;u,1)-\LL(-M,0;0,1)\leq \mathcal{A}_{\rm stat}(u)-\mathcal{A}_{\rm stat}(0)\\
\leq \LL(M+u,0;u,1)-\LL(M+u,0;0,1).
\end{multline}
Denote
\begin{equation}
\begin{aligned}
&X_-=\LL(-M,0;u,1)-\LL(-M,0;0,1),\\
&X=\mathcal{A}_{\rm stat}(u)-\mathcal{A}_{\rm stat}(0),\\
&X_+=\LL(M+u,0;u,1)-\LL(M+u,0;0,1),\\
&Y=\sqrt{2}B(u).
\end{aligned}
\end{equation}
Notice that in distribution
\begin{equation}\label{eqDistr}
\begin{aligned}
X_-&\stackrel{(d)}{=}\mathcal{A}_2(u)-(M+u)^2-\mathcal{A}_2(0)+M^2\\&=\mathcal{A}_2(u)-\mathcal{A}_2(0)-2 u M -u^2,\\
X&\stackrel{(d)}{=}\sqrt{2} B(u),\\
X_+&\stackrel{(d)}{=}\mathcal{A}_2(u)-M^2-\mathcal{A}_2(0)+(M+u)^2\\&=\mathcal{A}_2(u)-\mathcal{A}_2(0)+2 u M +u^2,
\end{aligned}
\end{equation}
where $\mathcal{A}_2$ is the Airy$_2$ process~\cite{Prahofer.Spohn-JSP2002}. More importantly, $X_-$ and $X_+$ are independent of $Y$.

Below we use the inequalities obtained by applying Cauchy-Schwarz and/or bounding the variance by the second moment, namely
\begin{equation}\label{eqIneq}
|\Cov(A\Id_C;B)|\leq \E(A^4)^{1/4} \Var(B)^{1/2} \Pb(C)^{1/4}
\end{equation}
for general random variables $A,B$ and events $C$.
First we decompose by linearity
\begin{equation}
\Cov(X;Y)=\Cov(X\Id_{G_M}; Y)+\Cov(X\Id_{G_M^c}; Y).
\end{equation}
Using \eqref{eqIneq} we get
\begin{equation}
|\Cov(X\Id_{G_M^c}; Y)|\leq \E(X^4)^{1/4} \Var(Y)^{1/2} \Pb(G_M^c)^{1/4}.
\end{equation}
By \eqref{eqDistr} we have $\E(X^4)=12 u^2$ and $\Var(Y)=2u$, so that, together with \eqref{eqGmC},
\begin{equation}
|\Cov(X\Id_{G_M^c}; Y)|\leq C u e^{-c M^3/4}.
\end{equation}

Next we need to bound $\Cov(X\Id_{G_M}; Y)$, which is equal to $\E(X \Id_{G_M} Y)$ since $\E(Y)=0$. We further decompose on $Y\geq 0$ and $Y<0$, and use \eqref{eqComparison}, to get
\begin{equation}
\begin{aligned}
\label{eq4}
&\Cov(X\Id_{G_M}; Y) \\
&= \E(X \Id_{G_M} Y \Id_{Y\geq 0})+\E(X \Id_{G_M} Y \Id_{Y<0})\\
&\leq\E(X_+ \Id_{G_M} Y \Id_{Y\geq 0})+\E(X_- \Id_{G_M} Y \Id_{Y<0})\\
&=\E(X_+  Y \Id_{Y\geq 0})+\E(X_- Y \Id_{Y<0}) \\
&\quad- \E(X_+ \Id_{G_M^c} Y \Id_{Y\geq 0})+\E(X_- \Id_{G_M^c} Y \Id_{Y<0}).
\end{aligned}
\end{equation}

The last two terms are bounded similarly. We have
\begin{equation}\label{eqS21}
\begin{aligned}
|\E(X_+ \Id_{G_M^c} Y \Id_{Y\geq 0})|&\leq \sqrt[4]{\E(X_+^4)} \sqrt[4]{\Pb(G_M^c)} \sqrt{\E(Y^2\Id_{Y\geq 0})}\\
&\leq \sqrt[4]{\E(X_+^4)} \sqrt[4]{\Pb(G_M^c)} \sqrt{\E(Y^2)}.
\end{aligned}
\end{equation}
Using $(a+b)^4\leq 8 (a^4+b^4)$ we get
\begin{equation}\label{eqS22}
\E(X_+^4) \leq 8 \E((\mathcal{A}_2(u)-\mathcal{A}_2(0))^4)+ 8 u^4 (u+2M)^4.
\end{equation}
By Lemma~\ref{lem2} below, in \eqref{eqS22} we get $\sqrt[4]{\E(X_+^4)}\leq C \sqrt{u} (1+u^2 M^4)^{1/4}$ for some constant $C$. So,
\begin{equation}
|\E(X_+ \Id_{G_M^c} Y \Id_{Y\geq 0})| \leq C u  (1+u^2 M^4)^{1/4} e^{-c M^3/4}.
\end{equation}
The same bound holds true for $\E(X_- \Id_{G_M^c} Y \Id_{Y<0})$.

The sum of the first two terms in  \eqref{eq4} are given by
\begin{equation}
\begin{aligned}
&\E(X_+  Y \Id_{Y\geq 0})+\E(X_- Y \Id_{Y<0}) \\
&= \E(X_- Y)+\E((X_+-X_-) Y \Id_{Y\geq 0})\\ 
&= (\E(X_+)-\E(X_-))\E(Y \Id_{Y\geq 0}),
\end{aligned}
\end{equation}
where we used the fact that $X_\pm$ and $Y$ are independent and that $\E(Y)=0$. We can compute explicitly $\E(Y \Id_{Y\geq 0})=\sqrt{u}/\sqrt{\pi}$ and $\E(X_+)-\E(X_-)=4 u M + 2 u^2$.
Thus,
\begin{equation}
\E(X_+  Y \Id_{Y\geq 0})+\E(X_- Y \Id_{Y<0}) = c_2 u^{3/2}(2M+u).
\end{equation}

Putting all together we get
\begin{multline}
|\Cov(\mathcal{A}_{\rm stat}(u)-\mathcal{A}_{\rm stat}(0);\sqrt{2} B(u))|\\
\leq C u \left(e^{-c M^3/4}(1+(1+u^2 M^4)^{1/4})+ u^{1/2}M\right).
\end{multline}
Finally we choose the value of $M$ depending on $u$. With $M=(\frac{2}{c} \ln(1/u))^{1/3}$ we get $(1+(1+u^2 M^4)^{1/4})=\Or(1)$, $u^{1/2}M=(2/c)^{1/3} \sqrt{u} (\ln(1/u))^{1/3}$ and $e^{-c M^3/4}=u^{1/2}$, which imply the claimed estimate.
\end{proof}

\begin{lemma}\label{lem2} There exists a constant $c_1$ such that for all $0\leq u\leq 1$,
\begin{equation}
\E((\mathcal{A}_2(u)-\mathcal{A}_2(0))^4)\leq c_1 u^2.
\end{equation}
\end{lemma}
\begin{proof}
Using the comparison inequality techniques of~\cite{CP15b}, as a consequence of the bounds of Lemma~3.4 of~\cite{FO17} (taking the $N\to\infty$ in there), we get
\begin{equation}\label{eq6}
|\mathcal{A}_2(u)-\mathcal{A}_2(0)|\leq \sqrt{2} B(u) + u^2+2\kappa u
\end{equation}
on a set $\Omega_\kappa$ with $\Pb(\Omega_\kappa)\geq 1-C e^{-c \kappa^2}$. Furthermore, let $\Omega_{\rm Cut}=\{ |\mathcal{A}_2(u)|\leq K\textrm{ and }|\mathcal{A}_2(0)|\leq K\}$. Then $\Pb(\Omega_{\rm Cut})\geq 1-C e^{-\frac43 K^{3/2}}$ since the one-point distribution of $\mathcal{A}_2(u)$ is the GUE Tracy-Widom distribution function~\cite{TW94}.

On $\Omega_G=\Omega_{\rm Cut}\cap \Omega_\kappa$, using $(a+b)^4\leq 8 (a^4+b^4)$, we get
\begin{multline}
\E((\mathcal{A}_2(u)-\mathcal{A}_2(0))^4) \leq 8 \E((\mathcal{A}_2(u)-\mathcal{A}_2(0))^4 \Id_G)\\
+8 \E((\mathcal{A}_2(u)-\mathcal{A}_2(0))^4\Id_{G}^c).
\end{multline}
The second term is bounded by $8 (2K)^4 \Pb(\Omega_G^c)$ since $|\mathcal{A}_2(u)-\mathcal{A}_2(0)|\leq 2K$. For the first term, using \eqref{eq6} we get
\begin{equation}
   \E((\mathcal{A}_2(u)-\mathcal{A}_2(0))^4 \Id_G)
   \leq 32 \E(|B(u)|^4)+ 8 (u^2+2\kappa u)^4.
\end{equation}
Since $\E(|B(u)|^4)=3 u^2$, if we choose for instance $K=u^{-2/3}$ and $\kappa=u^{-1/2}$, we obtain $\E((\mathcal{A}_2(u)-\mathcal{A}_2(0))^4)\leq C u^2$ for some constant $C>0$.
\end{proof}

Finally, let us compare the short-time behavior of the covariance of $\mathcal{A}_0$ with that of the other Airy processes. For the Airy$_2$ and the Airy$_1$ processes, the small $u$ behavior are the same: using the decomposition \eqref{eqS6} and the fact that locally the increments are as that of the stationary case, namely $\sqrt{2}B(u)$, one obtains~\cite{Prahofer.Spohn-JSP2002,QR12} $\Cov(\mathcal{A}_\ell(u),\mathcal{A}_\ell(0))=\Var(\mathcal{A}_\ell(0))-u+o(u)$, $\ell=1,2$.

For the stationary case, the decomposition \eqref{eqS6} and the fact that $\Var(\mathcal{A}_{\rm stat}(u)-\mathcal{A}_{\rm stat}(0))=2u$ leads to
\begin{multline}\label{eq5}
 \Cov(\mathcal{A}_{\rm stat}(0);\mathcal{A}_{\rm stat}(u)) 
=\Var(\mathcal{A}_{\rm stat}(0))-u\\
+\frac12 \left(\Var(\mathcal{A}_{\rm stat}(u))-\Var(\mathcal{A}_{\rm stat}(0))\right).
\end{multline}
Let $F_{{\rm BR},u}$ be the Baik-Rains distribution with parameter $u$~\cite{Baik.Rains-JSP2000,FS05a}. Then since $\mathcal{A}_{\rm stat}(u)$ is distributed according to $F_{{\rm BR},u}$ and it has expectation equal to zero (as a consequence of stationarity) we have
\begin{equation}
\Var(\mathcal{A}_{\rm stat}(u)) = \int_\mathbbm{R} s^2 dF_u(s) =:g_{\rm sc}(u).
\end{equation}
The latter is a scaling function which was already partially studied, see Section 7.2 of~\cite{PS01} for instance. $g_{\rm sc}(u)$ is symmetric, it increases linearly in $u$ as $|u|\to\infty$ (this cancels linear term in \eqref{eq5}) and $g''(0)>0$. This implies that for small $u$,
\begin{equation}
\Var(\mathcal{A}_{\rm stat}(u))-\Var(\mathcal{A}_{\rm stat}(0)) = g''(0) u^2 + \Or(u^4).
\end{equation}
Thus the covariance of the Airy$_{\rm stat}$, similarly to that of the Airy$_1$ and Airy$_2$ process, has a linear term $-u$ for small $u$, which differs from the linear term $-2u$ for the Airy$_0$ process.

\onecolumngrid
\section{Supplemental Text 3: Correlation functions $C_s(\ell, t)$ and $C_t(t_1, t_2)$ and physical observables in the quantum model}
\label{sec:TimeOrdering}

Since the current operators $\hat{J}_j(t)$ at different times do not commute, it may not be obvious whether the correlation functions related to the magnetization transfer $\Delta \hat{S}_j = \int_0^t \hat{J}_j(t') dt'$ are measurable or not. However, since these correlation functions are obtained from the expectation value of the current operators, they are measurable, at least in principle, as we explicitly describe below.

Since $\langle \Delta \hat{S}_j(t) \rangle = 0$, the correlation functions are given as
\begin{equation}
        C_s(\ell, t) = \langle \Delta\hat{S}_{0} (t)\Delta\hat{S}_{\ell}(t)\rangle,  \quad C_t(t_1, t_2) = \langle \Delta\hat{S}_{0}(t_1) \Delta\hat{S}_{0}(t_2)\rangle,  \label{eq:QuantCorr}
\end{equation}
 where the brackets represent the expectation values for the infinite-temperature equilibrium state.
Using the current-current correlation function $F_\ell(s) \equiv 2^{-L} \mathrm{Tr}[\hat{J}_\ell e^{-i\hat{H}s}\hat{J}_0 e^{i\hat{H}s}]$, they are recast into
\begin{equation}
        C_s(\ell, t) = \int_{0}^t \!\!\! dt' \!\! \int_{0}^t \!\!\! dt'' F_\ell(t'- t''),  \quad C_t(t_1, t_2) = \int_{0}^{t_1} \!\!\! dt' \!\! \int_{0}^{t_2} \!\!\! dt'' F_\ell(t'- t''),
\end{equation}
and thus the problem is reduced to evaluating $F_\ell(s)$. 

While there are several ways to obtain $F_\ell(s)$, here we consider an approach based on the real-time evolution of a density matrix. 
Let us consider $G_\ell(s)$, defined as follows:
\begin{equation}
    G_\ell(s) \equiv \mathrm{Tr}[\hat{J}_\ell \hat{\rho}(s)] 
\end{equation}
where $\hat{\rho}(s) \equiv e^{-i\hat{H}s}\hat{\rho}_0 e^{i\hat{H}s}$, $\hat{\rho}_0 \equiv (\hat{I} - 2\hat{J}_0)/2^{L}$, $\hat{I}\equiv\bigotimes_{j=-L/2+1}^{L/2} \hat{I}_j$, and $\hat{I}_j$ is the identity operator acting on the $j$-th site. 
The initial density matrix 
$\hat{\rho}_0$ is also represented as 
\begin{equation}
    \hat{\rho}_0 =  \left(\bigotimes_{j=-\frac{L}{2}+1}^{-1} \frac{\hat{I}_j}{2} \right) \otimes \hat{\rho}_{01} \otimes  \left(\bigotimes_{j=2}^{\frac{L}{2}} \frac{\hat{I}_j}{2} \right), \quad 
    \hat{\rho}_{01} = \frac{1}{4} \ket{\uparrow \uparrow}\!\!\bra{\uparrow \uparrow} + \frac{1}{2} \ket{\psi}\!\! \bra{\psi} + \frac{1}{4} \ket{\downarrow \downarrow}\!\! \bra{\downarrow \downarrow},
\end{equation}
with $\ket{\psi} = (\ket{\uparrow \downarrow} + i \ket{\downarrow \uparrow})/\sqrt{2}$. Here, $\ket{\sigma \sigma'} \equiv \ket{\sigma}_0 \otimes \ket{\sigma'}_1$ and $\ket{\sigma}_j$ is the eigenstate of $\hat{S}_j^z$, i.e., $\hat{S}_j^z \ket{\uparrow}_j = (1/2) \ket{\uparrow}_j$ and $\hat{S}_j^z \ket{\downarrow}_j = (-1/2) \ket{\downarrow}_j$.
Note that $\hat{\rho}_0$ is different (at the 0th and 1st sites) from the infinite-temperature equilibrium state that we used for simulations as the initial state, but with this initial state, we can obtain the correlation functions \pref{eq:QuantCorr} studied in this work.
Specifically, since $\mathrm{Tr}[\hat{J}_\ell]=0$, the function $G_\ell(s)$ is related to $F_\ell(s)$ as 
\begin{equation}
    G_\ell(s) = -2F_\ell(s).
\end{equation}
This shows that $F_\ell(s)$ is obtained from the current expectation value $G_\ell(t)$. This means that we can compute $C_s(\ell,t)$ and $C_t(t_1, t_2)$ from the current measurement in the time evolution starting from the initial density matrix $\hat{\rho}_0$.

\onecolumngrid
\newpage
\section{Supplemental Figure}

\begin{figure}[h!]
\includegraphics[width=\hsize,clip]{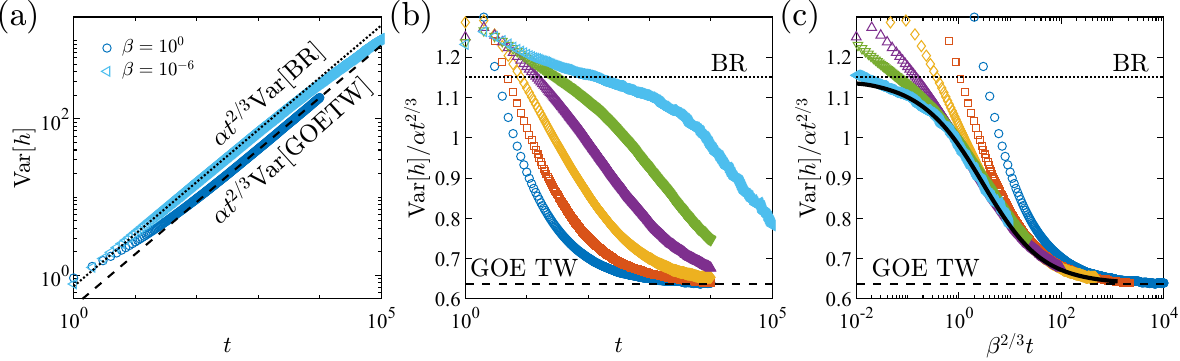}
\centering
\caption{
Stationary-to-flat crossover in TASEP with Ornstein-Uhlenbeck initial conditions.
The initial conditions $h_0(x,0)$ were generated by Monte Carlo sampling with statistical weight $e^{-\beta \int h_0(x,0)^2 dx}$.
The parameter $\alpha$ for this case is known to be $2^{-2/3}$ exactly.
(a) Height variance Var$[h]$ against time $t$.
The dotted and dashed lines display the power laws for the stationary case (Baik-Rains distribution) and the flat case (GOE Tracy-Widom distribution), respectively, proportional to $t^{2/3}$ for both cases.
(b)(c) Rescaled variance Var$[h] / \alpha t^{2/3}$ against $t$ (b) and $\beta^{2/3}t$ (c), for $\beta = 10^0, 10^{-1}, 10^{-2}, 10^{-3}, 10^{-4}, 10^{-6}$ from left to right in (b).
The thick solid line in (c) is the flat-to-stationary crossover function obtained in Ref.\,\cite{Takeuchi-PRL2013}, displayed with an arbitrary horizontal shift.
}
\label{figS:TASEP}
\end{figure}

\bibliography{ref}